\documentclass[10pt,conference]{IEEEtran}

\makeatletter
\def\ps@headings{%
\def\@oddhead{\mbox{}\scriptsize\rightmark \hfil \thepage}%
\def\@evenhead{\scriptsize\thepage \hfil \leftmark\mbox{}}%
\def\@oddfoot{}%
\def\@evenfoot{}}
\makeatother
\usepackage{paralist}

\usepackage[ruled,vlined]{algorithm2e}

\usepackage{amssymb}
\usepackage{amsmath} 

\usepackage{array}
\usepackage{graphicx}
\usepackage{epstopdf}
\usepackage{stmaryrd}
\usepackage{dsfont}
\usepackage{subfig}
\usepackage{paralist}
\usepackage{verbatim}

\newif\iflong 
\longfalse

\newif\ifcomm
\commtrue

\newtheorem{thm}{Theorem}
\newtheorem{cor}{Corollary}
\newtheorem{lem}{Lemma}

\newcommand{\set}[1]{\left\{#1\right\}}

\newenvironment{proof-sketch}{{\noindent\em Proof Sketch.}\hspace*{0.3em}}{\qed\medskip}
\newenvironment{proof}{{\bf Proof.}}{\hfill\rule{2mm}{2mm}\\}
\newenvironment{proof-of}[1]{{\noindent\em Proof of #1.}\hspace*{0.3em}}{\qed\medskip}

\newcounter{assumption}
\newcommand{\theassumptionletter}{A}
\renewcommand{\theassumption}{\theassumptionletter\arabic{assumption}}

\newcommand{\MF}{\mathcal{F}}

\newcommand{\MRR}{\mathcal{R}}

\newcommand{\MK}{\mathcal{K}}
\newcommand{\MP}{\mathcal{P}}

\newcommand{\beq}{\begin{equation}}
\newcommand{\eeq}{\end{equation}}
\newcommand{\beqa}{\begin{eqnarray}}
\newcommand{\eeqa}{\end{eqnarray}}
\newcommand{\beqan}{\begin{eqnarray*}}
\newcommand{\eeqan}{\end{eqnarray*}}
\newcommand{\ben}{\begin{eqnarray*}}
\newcommand{\een}{\end{eqnarray*}}

\ifcomm
   \newcommand\comm[1]{\textcolor{blue}{ #1}}
\else
   \newcommand\comm[1]{}
   \renewcommand{\todo}[1]{}
\fi

\newcommand{\bm}{\boldmath}

\newcommand{\qed}{{\hfill$\Box$}}
\newcommand{\C}{\mbox{\bm $C$}}

\newcommand{\D}{\mbox{\bm $D$}}

\newcommand{\Gg}{\mbox{\bm $g$}}

\newcommand{\xx}{\mbox{$\mathbf{x}$}}

\newcommand{\yy}{\mbox{$\mathbf y$}}
\newcommand{\G}{\mbox{\bm $G$}}

\newcommand{\A}{\mbox{\bm $A$}}
\newcommand{\B}{\mbox{\bm $B$}}

\newcommand{\Y}{\mbox{\bm $Y$}}
\newcommand{\X}{\mbox{\bm $X$}}

\makeatletter
\renewcommand\paragraph{\@startsection{paragraph}{4}{\z@}%
    {1.5ex plus .2ex minus .3ex}%
            {-0em}%
                        {\normalsize\bf}}

\newcommand{\captionfonts}{\small}

\long\def\@makecaption#1#2{%
  \vskip\abovecaptionskip
  \sbox\@tempboxa{{\captionfonts #1: #2}}%
  \ifdim \wd\@tempboxa >\hsize
    {\captionfonts #1: #2\par}
  \else
    \hbox to\hsize{\hfil\box\@tempboxa\hfil}%
  \fi
  \vskip\belowcaptionskip}
\makeatother

\begin{document}
\title{Binary Independent Component Analysis with OR Mixtures}
\author{
\begin{tabular}{c}
Huy Nguyen and Rong Zheng  \\
Department of Computer Science \\
University of Houston \\
Houston, TX 77204 \\
E-mail: {\it nahuy@cs.uh.edu, rzheng@cs.uh.edu}
\end{tabular}
}

\maketitle

\begin{abstract}
Independent component analysis (ICA) is a computational method for
separating a multivariate signal into subcomponents assuming the
mutual statistical independence of the non-Gaussian source
signals. The classical Independent Components Analysis (ICA)
framework usually assumes linear combinations of independent
sources over the field of real-valued numbers $\MRR$.  In this
paper, we investigate binary ICA for OR mixtures (bICA), which can
find applications in many domains including medical diagnosis,
multi-cluster assignment, Internet tomography and network resource
management. We prove that bICA is uniquely identifiable under the
disjunctive generation model, and propose a deterministic
iterative algorithm to determine the distribution of the latent
random variables and the mixing matrix. The inverse problem concerning
inferring the values of latent variables are also considered along with
noisy measurements. We conduct an extensive simulation study  to
verify the effectiveness of the propose algorithm and present examples of
real-world applications where bICA can be applied.\\
\end{abstract}

\begin{IEEEkeywords}
Boolean functions, clustering methods, graph matching, independent component analysis.
\end{IEEEkeywords}

\section{Introduction}
\label{sect:intro}
Independent component analysis (ICA) is a computational method for separating a
multivariate signal into additive subcomponents supposing the mutual
statistical independence of the non-Gaussian source signals. The classical
Independent Components Analysis (ICA) framework usually assumes linear
combinations of independent sources over the field of real-valued numbers
$\MRR$. Consider the following generative data model where the observations are
disjunctive mixtures of binary independent sources. Let $\xx = [x_1, x_2, \ldots, x_n]^T$ be a $m$-dimension binary
random vector with joint distribution $\MP(\xx)$, which are observable.  $\xx$
is generated from a set of $n$ independent binary random variables $\yy
= [y_1,y_2,\ldots,y_n]^T$ as follows,
\beq
x_i = \bigvee_{j=1}^{n}{(g_{ij}\wedge y_j)}, \mbox{ $i = 1, \ldots, m$},
\label{eq:boolean}
\eeq
where $\wedge$ is Boolean {\em AND}, $\vee$ is Boolean {\em OR}, and $g_{ij}$ is
the entry in the $i$'th row and $j$'th column of an unknown binary mixing
matrix $\G$. Throughout this paper, we denote by $\G_{i,:}$ and $\G_{:,j}$
the $i$'th row and $j$'th column of matrix $\G$ respectively. For the ease of
presentation, we introduce a short-hand notation for the above disjunctive model as,
$$
\xx = \G \otimes \yy.
$$

The relationship between observable variables in $\xx$ and latent binary
variables in $\yy$ can also be represented by an undirected bi-partite graph
$G=(U,V,E)$, where $U = \{x_1, x_2, \ldots, x_m\}$ and $V = \{y_1, y_2, \ldots,
y_n\}$ (Figure~\ref{fig:model}). An edge $e=(x_i,y_j)$ exists  if $g_{ij} = 1$. We will refer to
$\G$ as the binary adjacency matrix of graph $G$. The key notations used in this
paper are listed in Table~\ref{tab:mainnotations}.
\begin{figure}[tp]
\begin{center}
\includegraphics[width=3in]{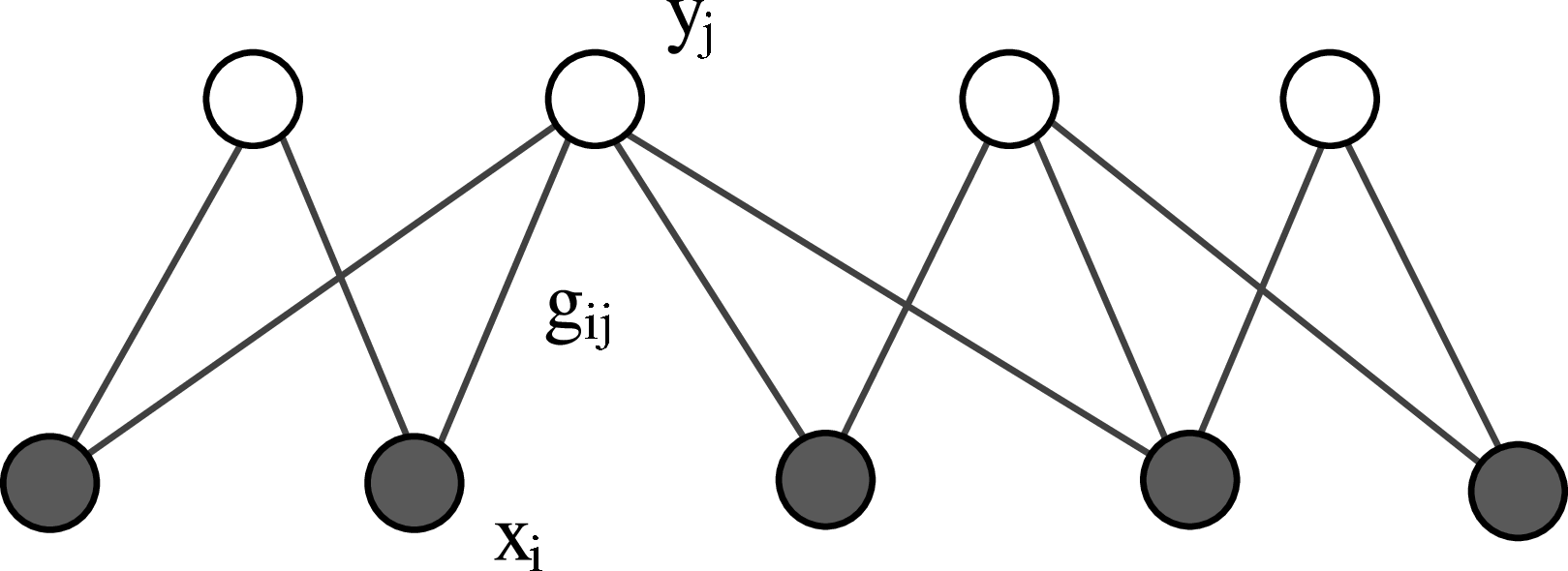}
\end{center}
\caption{Illustration of the OR mixture model}
\label{fig:model}
\end{figure}

Consider an $m\times T$ matrix $\X$ and an $n\times T$ matrix $\Y$, which are the collection of $T$
realizations of random vector $\xx$ and $\yy$ respectively. The goal of {\it binary independent component
analysis with OR mixtures (bICA)} is to estimate the distribution of the
latency random variables $\yy$ and the binary mixing matrix $\G$ from $\X$ such
that $\X$ can be decomposed into OR mixtures of columns of $\Y$.

In this paper, we make the following contributions.  First, we investigate the
identifiability of bICA and prove that under the disjunctive generation model,
the OR mixing is identifiable up to permutation. Second, we develop an iterative
algorithm for bICA that does not make assumptions on the noise model or
prior distributions of the mixing matrix. Interestingly, the approach is shown to be
robust under moderate to medium XOR noises and insufficient samples. We
furthermore consider the inverse problem of inferring the values of $\yy$ given
noisy observations $\X$ and the inferred model.  Finally, we present two case
studies to illustrate how bICA can be used to model and solve real-world problems. .

\begin{table}[tp]
\caption{Notations}
\centering 
\vspace{0.1in}
{\normalsize
\begin{tabular}{c | c}
\hline\hline
& the number of observable variables, \\[-0.8ex]
\raisebox{1.2ex}{$m,n,T$}
& hidden variables and observations \\
\hline
& the bi-partite graph representing the\\[-0.8ex]
\raisebox{1.2ex}{$G$}
& observable-hidden variable relationship  \\
\hline
& the vector of $m$ binary observations\\[-0.8ex]
\raisebox{1.2ex}{$\xx_{m\times 1}$}
& from $m$ observable variables\\
\hline
& the vector of $n$ binary activities\\[-0.8ex]
\raisebox{1.2ex}{$\yy_{n\times 1}$}
& from $n$ hidden variables \\
\hline
& the collection of $T$ realizations of $\xx$ \\[-0.8ex]
\raisebox{1.2ex}{$\X_{m\times T}$}
& (observation matrix) \\
\hline
& the collection of $T$ realizations of $\yy$ \\[-0.8ex]
\raisebox{1.2ex}{$\Y_{n\times T}$}
& (activity matrix) \\
\hline
$\G_{m \times n}$ & the binary adjacency matrix of graph $G$ \\
\hline
& the probability vector associated with \\[-0.8ex]
\raisebox{1.2ex}{$p_{1\times n}$}
& $n$ (Bernoulli) hidden variables \\
\hline\hline
\end{tabular}
}
\label{tab:mainnotations}
\end{table}

The rest of the paper is organized as follows. In Section~\ref{sec:related}, we
give a brief overview of related work. In Section~\ref{sec:property}, several
important properties of bICA are proved. In Section~\ref{sec:algo}, we elaborate
on an iterative procedure to infer bICA and several complexity reduction
techniques. Formulation and solution to the inverse problem with noisy
measurements are presented in Section~\ref{sec:inverse}. Effectiveness of the
proposed method is evaluated through simulation results in Section~\ref{sec:eval}.
Real-world problem domains where bICA can be applied are discussed in Section~\ref{sec:application}
and followed by conclusion and future work in Section~\ref{sec:conclusion}.

\section{Related work}
\label{sec:related}
Most ICA methods assume linear mixing of continuous signals~\cite{Hyvarinen00}.
A special variant of ICA, called binary ICA (BICA), considers boolean mixing
(e.g., OR, XOR etc.) of binary signals. Existing solutions to BICA mainly
differ in their assumptions of the binary operator (e.g., OR or XOR), the prior
distribution of the mixing matrix, noise model, and/or hidden causes.

In \cite{Yeredor07}, Yeredor considers BICA in XOR mixtures and investigates
the identifiability problem. A deflation algorithm is proposed for source
separation based on entropy minimization. Since XOR is addition in GF(2), BICA
in XOR mixtures can be viewed as the binary counterpart of classical linear ICA
problems. In \cite{Yeredor07}, the number of independent random sources $K$ is
assumed to be known. Furthermore, the mixing matrix is a $K$-by-$K$ invertible
matrix.  Under these constraints, it has been proved that the XOR model is
invertible and there exist a unique transformation matrix to recover the
independent components up to permutation ambiguity. Though our proof of
identifiability in this paper is inspired by the approach in \cite{Yeredor07},
due to the ``non-linearity" of OR operations, the notion of invertible matrices
no longer apply. New proofs and algorithms are warranted to unravel the
properties of binary OR mixtures.

In \cite{Kabán_factorisationand}, the problem of factorization and de-noise of
binary data due to independent continuous sources is considered.  The sources
are assumed to be continuous following beta distribution in [0, 1]. Conditional
on the latent variables, the observations follow the independent Bernoulli
likelihood model with mean vectors taking the form of a linear mixture of the
latent variables. The mixing coefficients are assumed to non-negative and sum
to one.  A variational EM solution is devised to infer the mixing coefficients.
A post-process step is applied to quantize the recovered ``gray-scale" sources
into binary ones.  While the mixing model in \cite{Kabán_factorisationand} can
find many real world applications, it is not suitable in the case of OR
mixtures.

In \cite{Singliar2006}, a noise-OR model is introduced to model dependency
among observable random variables using $K$ (known) latent factors.  A variational inference
algorithm is developed. In the noise-OR model, the probabilistic dependency
between observable vectors and latent vectors is modeled via the noise-OR
conditional distribution. The dimension of the latent vector is assumed to be
known and less than that of the observable.

In \cite{computer06anon-parametric}, Wood {\it et al.} consider the problem of
inferring infinite number of hidden causes following the same Bernoulli
distribution. Observations are generated from a noise-OR distribution.  Prior
of the infinite mixing matrix is modeled as the Indian buffet process~\cite{Griffiths05infinitelatent}.
Reversible jump Markov chain Monte Carlo and  Gibbs sampler techniques are
applied to determine the mixing matrix based on observations. In our model, the
hidden causes are finite in size, and may follow different distribution.
Streith {\it et al.}~\cite{Streich09} study the problem of multi-assignment
clustering for boolean data, where the observations are either drawn from a
signal following OR mixtures or from a noise component.  The key assumption
made in the work is that the elements of matrix $\X$ are conditionally
independent given the model parameters (as opposed to the latent variables).
This greatly reduces the computational complexity and makes the scheme amenable
to a gradient descent-based optimization solution. However, this assumption is
in general invalid.

There exists a large body of work on blind deconvolution with binary sources in
the context of wireless communication~\cite{Diamantaras06,LiCichocki03}. In
time-invariant linear channels, the output signal $x(k)$ is a convolution of
the channel realizations $a(k)$ and the input signal $s(k)$, $k=1,2,\ldots, K$
as follows:
\beq
x(k) =
\sum_{l=0}^{L}{a(l)s(k-l)}, k=1,\ldots, K.
\label{eq:deconvolution}
\eeq
The objective is to recover the input signal $s$.  Both stochastic and
deterministic approaches have been devised for blind deconvolution.  As evident
from \eqref{eq:deconvolution}, the output signals are linear mixtures of the
input sources in time, and additionally the mixture model follows a specific
structure.

\begin{table*}[tp]
\begin{center}
\begin{tabular}{|c||c|c|c|c|}
\hline
Algorithm  & Sources & Generative model & Under/Over determined & Dimension of latent variables\\
\hline \hline
\cite{Yeredor07} & Binary & Binary XOR  & --  & Known \\
\hline
\cite{Kabán_factorisationand} & Continuous & Linear & Over & Known \\
\hline
\cite{Singliar2006} & Binary & Noise-OR  & Over & Known \\
\hline
\cite{computer06anon-parametric} & Binary & Noise-OR & Under & Infinite \\
\hline
\cite{Streich09} & Binary & Binary OR & Over & Known \\
\hline
\cite{Frolov07,Belohlavek2010} & Binary & Binary OR & Over & Unknown, try to minimize \\
\hline
\cite{Diamantaras06,LiCichocki03} & Binary & Linear & -- & Known \\
\hline\hline
{\bf bICA} & {\bf Binary} & {\bf Binary OR} & {\bf Under} & {\bf Unknown but finite} \\
\hline
\end{tabular}
\caption{Related work comparison chart}
\label{tab:comparison}
\end{center}
\end{table*}

Literature on boolean/binary factor analysis (BFA) is also related to our work.
The goal of BFA is to decompose a binary matrix $\X_{m\times T}$ into $\A_{m
\times n} \otimes \B_{n\times T}$ with $\otimes$ being the \textit{OR} mixture
relationship as defined in (\ref{eq:boolean}). $\X$ in BFA is
often called an attribute-object matrix providing $m$-dimension attributes of
$T$ objects.  $\A$ and $\B$ are  the attribute-factor and factor-object
matrices. All the elements in $\X$, $\A$, and $\B$ are either 0 or 1. $n$ is
defined to be the number of underlying factors and is assumed to be
considerably smaller than the number of objects $T$. BFA methods aim to find a
feasible decomposition minimizing $n$.  Frolov {\it et al.} study the problem
of factoring a binary matrix using Hopfield neural networks~\cite{Frolov05,
Frolov07, Frolov07_2}. This approach is based on heuristic and do not provide
much  theoretical insight regarding the properties of the resulting
decomposition.  More recently, Belohlavek {\it et al.} propose a matrix
decomposition method utilizing formal concept analysis~\cite{Belohlavek2010}.
The paper claims that optimal decomposition with the minimum number of factors
are those where factors are formal concepts.  It is important to note that even
though BFA assumes the same disjunctive mixture model as in our work, the
objective is different.  While BFA tries to find a matrix factorization so that
the number of factors are minimized, bICA tries to identify independent
components. One can easily come up an example, where the number of independent
components (factors) is larger than the number of attributes. Since BFA always
finds factors no larger than the number of attributes, the resulting factors
are clearly dependent in this case.

Finally, \cite{computer06anon-parametric} consider the under-presented case of
less observations than latent sources with continuous noise, while
\cite{Kabán_factorisationand,Streich09,Frolov07,Belohlavek2010} deal with the over-determined case,
where the number of observation variables are much larger. In this work, we consider
primarily the under-presented cases that we typically encounter in data
networks where the number of sensors are much smaller and the number of signal
sources (i.e. users).

We summarize the aforementioned related work in Table~\ref{tab:comparison}.

\section{Properties of bICA}
\label{sec:property}
In this section, we investigate the fundamental properties of bICA. In
particular, we are interested in the following questions:
\begin{itemize}
\item {\bf Expressiveness:} can any set of binary random variables be
decomposed into binary independent components using OR mixtures?
\item {\bf Independence of OR mixtures:} for mixtures of independent sources,
what is the condition that they are independent?
\item {\bf Identifiability:} given a set of  binary random variables following
the bICA data model, is the decomposition unique?
\end{itemize}
\paragraph*{Expressiveness}
Expressiveness of OR mixtures is limited. This can be shown through an example.
Let $y_1$ and $y_2$ be two independent binary random variables with $P(y_1=1) =
p \neq 0.5$ and $P(y_2=1) = q \neq 0.5$. Let $x_1 = y_1$ and $x_2 = y_1 + y_2$,
where `+' is addition in the finite field GF(2). It is easy to see that $x_1$
and $x_2$ are correlated since $P(x_2 = 1) = P(y_1 = 1)P(y_2 = 0) + P(y_1 =
0)P(y_2 = 1) = q(1-p)+p(1-q)$, $P(x_1 = 1) = p$, while $P(x_1 = 1, x_2 = 1) =
P(y_2 = 0) = 1-q$.  On the other hand, $x_2$ can not be decomposed into an OR
mixture of $y_1$ and $y_2$. This essentially shows that OR mixtures of binary
random variables only span a subset of multi-variate binary distributions.
There exist correlated binary random variables ($x_1, x_2$ in this example)
that cannot be modeled as OR mixtures of independent binary components.
\paragraph*{Independence of mixtures}
Now we turn to the second question, namely, under what condition are binary
random variables that follow the OR mixture model independent.  In general, pairwise
independent random variables are not jointly independent.  Interestingly, we
show that pairwise independence implies joint independence for OR mixtures.

\begin{thm} Let $\yy = [y_1, y_2, \ldots, y_n]^T$ denote $n$ statistically
independent sources in $GF(2)$, the $i$-th source having 1-probability $p_i$ .
Let $\xx = \D \otimes \yy$, where $\D$ is a $m\times n$ matrix (with elements in
GF(2)). Let $\eta(\xx)$ and $\C(\xx)$ denote the mean and covariance (resp.) of $\xx$. If:
\begin{enumerate}
\item All elements of $\eta(\xx)$ are nonzero and not 1's (called {\it
non-degenerate}),
\item $\C(\xx)$ is diagonal,
\end{enumerate}
Then i) $m = n$, and ii) $\D$ is a permutation matrix.
\label{thm:perm}
\end{thm}

Let us first establish the following lemmas.
\begin{lem}
Let $u$ and $v$ be two RVs in GF(2) with 1-probabilities $p_u$ and $p_v$ (resp.),
and let $w \stackrel{\Delta}{=}  u \vee v$ If $u$ and $v$ are independent,
non-degenerate ($0 < p, q < 1$) then $w$ is also non-degenerate.
\label{lem:1}
\end{lem}
\begin{proof}
Clearly, $p_w =  P(w = 1) = 1 - (1-p_u)(1-p_v)$. Since $0 < p_u, p_v < 1$, we have $0 < p_w < 1$.
\end{proof}
\begin{lem}
Consider non-degenerate independent binary random variables $y_1, y_2, y_3$.
Then, $x_1 = y_1\vee y_2$ and $x_2 = y_3$ are independent.
\label{lem:2}
\end{lem}
\begin{proof}
To prove independence of two binary random variables $x_1$ and $x_2$, it is
sufficient to show $P(x_1 = 1, x_2 = 1) = P(x_1 = 1)P(x_2 = 1)$.
\begin{equation}
\begin{array}{lll}
P(x_1 = 1, x_2 = 1) & = & P(y_1 = 1, y_2 = 1, y_3 = 1) \\
            & + & P(y_1 = 1, y_2 = 0, y_3 = 1) \\
            & + & P(y_1 = 0, y_2 = 1, y_3 = 1) \\
            & = & P(y_1 = 1)P(y_2 = 1)P(y_3 = 1) \\
            & + & P(y_1 = 1)P(y_2 = 0)P(y_3 = 1) \\
            & + & P(y_1 = 0)P(y_2 = 1)P(y_3 = 1) \\
            & = & P(x_1 = 1)P(x_2 = 1)
\end{array}
\end{equation}
\end{proof}
Similar, we can show the following result.
\begin{lem}
Consider non-degenerate independent binary random variables $y_1, y_2, y_3$.
Then, $x_1 = y_1\vee y_2$ and $x_2 = y_1\vee y_3$ are correlated.
\label{lem:corr}
\end{lem}

Now we are in the position to prove Theorem~\ref{thm:perm}.
\begin{proof}[Proof of Theorem~\ref{thm:perm}]
We prove by contradiction. The essence of the proof is similar to that in \cite{Yeredor07}.
Let us assume now that $\D$ is a general matrix, and consider any pair $x_k$ and
$x_l$ ($k \neq l$) in $\xx$. $x_k$ and $x_l$ are OR mixtures of respective subgroups of the
sources, indexed by the 1-s in $\D_{k,:}$, and $\D_{l,:}$, the $k$-th and $l$-th rows (respectively) of
$\D$. These two subgroups consist of, in turn, three other subgroups (some of which
may be empty):
\begin{enumerate}
\item Sub-group 1: Sources common to $\D_{k,:}$ and $\D_{l,:}$ . Denote the OR mixing of
these sources as $u$;
\item Sub-group 2: Sources included in $\D_{k,:}$ but excluded from $\D_{l,:}$. Denote the
OR mixing of these sources as $v_1$;
\item Sub-group 3: Sources included in $\D_{l,:}$ but excluded from $\D_{k,:}$. Denote the
OR mixing of these sources as $v_2$.
\end{enumerate}
In other words, $x_k = u \vee v_1$ and $x_l = u \vee v_2$. By applying Lemma~\ref{lem:2}
iteratively, we can show that $v_1$ and $v_2$ are independent and
non-degenerate. Furthermore, if $u \neq 0$, then $u$ is independent of $v_1$
and $v_2$. From Lemma~\ref{lem:corr}, we show that $x_k$ and $x_l$ are correlated. This
contradicts with the condition that  $\C(\xx)$  is diagonal. This implies that $u
= 0$. Therefore, the two rows $\D_{k,:}$ and $\D_{l,:}$ do not share common
sources, or, in other words, there is no column $j$ in $\D$ such that both
$\D_{k,j}$ and $\D_{l,j}$ are both 1. There are only $m$ such columns. Thus,
$m=n$.  Furthermore, $\D$ is a permutation matrix.
\end{proof}

Theorem~\ref{thm:perm} necessarily implies the following result:
\begin{cor}
Let $\xx = \G \otimes \yy$ for some $\G$ and independent non-degenerate sources $\yy$. Then, if
elements of $\xx$ is non-degenerate and pair-wise independent, the elements in
$\xx$ are jointly independent.
\end{cor}

\paragraph*{Identifiability}
Let $\xx= [x_1, \ldots, x_m]^T$.
Define the set
$$Y(\xx) = \{\yy\mid \bigvee_{j=1}^{n}{(g_{ij}\wedge y_j)} = x_i, \forall i=1, \ldots, m\}.$$
Therefore,
\beq
\begin{array}{lll}
\MP(\xx) & = & \MP(\yy \in Y(\xx)) = \sum_{\yy \in Y(\xx)} \MP(\yy) \\
         & = & \sum_{\yy \in Y(\xx)}{\prod_{i=1}^n{p_i^{y_i}(1-p_i)^{1-y_i}}}
\end{array}
\label{eq:relation}
\eeq
where $\MP(\yy)$ is the joint probability of $\yy$, and  $p_i
\stackrel{\Delta}{=} \MP(y_i = 1)$. The last equality is due to the
independence among $y_i$'s.

To see whether $\yy$ is uniquely identifiable from $\xx$, we first restrict $\G$ such that it has no identical columns,
namely, each $y_j$ contributes to  a unique set of $x_i$'s. Otherwise, if
$\G_{:,i}$ and $\G_{:,j}$ are identical, we can merge $y_i$ and $y_j$ by a new
component corresponding to $y_i \vee y_j$.  Under the restriction, we can
initialize $n = 2^m-1$ and $\G$ of dimension $m\times 2^m-1$ with rows
being all possible $n$ binary values.  The $\G$ matrix corresponds to  a
complete bipartite graph, where an edge exists between any two vertices in $U$
and $V$, respectively. For a random variable $y_j\in V$, its neighbors in $U$
is given by the non-zero entries in $\G_{:,j}$.
Thus, at most $2^m - 1$ independent components can be identified.
Given the distribution of random variables $\xx\in \{0,1\}^m$,  $2^m - 1$
equations can be obtained from \eqref{eq:relation}. As there are at most
$2^m -1$ unknowns (i.e., $p_i, i = 1,\ldots, n$), the probability of $y_j$ can
be determined if a solution exists.  To see the solution uniquely exists, we
present a constructive proof as follows.

Let $\Gg_k$, $k = 1, \ldots, 2^m-1$ be a $m$-dimension binary column vector, and the degree of
$\Gg_k$, $d(\Gg_k)$ is the number of ones in $\Gg_k$. Define the frequency function $\MF_k = \MP(\xx = \Gg_{k}) = \MP(x_i
= g_{ik}, i = 1, \ldots m)$. For each $\Gg_k$, we associate it with an
independent component $y_k$. The goal is to show that $p_k \stackrel{\Delta}{=}
\MP(y_k = 1)$ can be uniquely decided. Starting from $\Gg_{k}$ with the lowest
degree, the derivation proceeds to determine $p_k$ with increasing degree in
$\Gg_{k}$.

\noindent{\bf \textit{Basis}:} It is easy to show that $\MF_0 =
\prod_{j=1}^{2^m-1}{(1-p_j)}$. Since $p_k$'s are non-degenerate, $\MF_0 > 0$.
For $k$, s.t., $d(\Gg_k) = 1$, we have
$$
\MF_k = p_k\prod_{j=1, j \neq k}^{2^m-1}{(1-p_j)}
$$
Therefore,
\beq
\label{eq:1case}
p_k = \frac{\MF_k}{\MF_k + \MF_0}  \MF_0
\eeq

\noindent{\bf \textit{Induction}:}
Define $\Gg_i \prec  \Gg_j$ if $\Gg_i \neq  \Gg_j$, and
$\forall l$, s.t., $g_{li} = 1$, $g_{lj} = 1$.
Let $S_k$ be the set of indices $i$'s, s.t., $\MF_i \neq 0$ and $\Gg_i \prec \Gg_k$,
$\forall i \in S$.
If $S_k = \emptyset$, then \eqref{eq:1case} applies. Otherwise, we have
$$
\begin{array}{lll}
\displaystyle
& \MF_k \\ = & \prod_{j \not\in S_k, j\neq k}{(1-p_j)}\times (p_k + \\ & (1-p_k)\sum_{B \subset S_k, \bigvee_{i\in B}{\Gg_i} = \Gg_k}{\prod_{i\in B}{p_i}\prod_{i\in B-S_k}{(1-p_i)}}) \\
= & \frac{\MF_0}{(1-p_k)\prod_{j \in S_k}{(1-p_j)}}\times (p_k + \\ & (1-p_k)\sum_{B \subset S_k, \bigvee_{i\in B}{\Gg_i} = \Gg_k}{\prod_{i\in B}{p_i}\prod_{i\in B-S_k}{(1-p_i)}}) \\
= & \frac{\MF_0}{\prod_{j \in S_k}{(1-p_j)}}\times (\frac{p_k}{1-p_k} + \\ & \sum_{B \subset S_k, \bigvee_{i\in B}{\Gg_i} = \Gg_k}{\prod_{i\in B}{p_i}\prod_{i\in B-S_k}{(1-p_i)}})
\end{array}
$$
where $\bigvee_{i\in B}{\Gg_i}$ indicates the entry-wise OR of $\Gg_i$'s for $i \in B$. Let us define $L_k \stackrel{\Delta}{=} \sum_{B \subset S, \bigvee_{i\in B}{\Gg_i} = \Gg_k}{\prod_{i\in B}{p_i}\prod_{i\in B-S}{(1- p_i)}})$. Then,

\beq
\label{eq:2case}
p_k = \frac{\MF_k\prod_{i\in S_k}{(1-p_i)} - {\MF_0}L_k}{\MF_0 + \MF_k\prod_{i\in S_k}{(1-p_i)}- \MF_0 L_k}.
\eeq
%
It is easy to verify that when the $y_i$'s are non-degenerate, all the
denominators are positive. This proves that a solution to \eqref{eq:relation}
exists and is unique. However, direct application of the construction suffers
from several problems. First, all $\MF_{k}$'s need to be computed from the
data, which requires a large amount of observations.  Second, the property that
$\MF_0 \neq 0$ is very critical in estimating $p_k$'s. When $\MF_0$ is small,
it cannot be estimated reliably.  Third, enumerating $S_k$ for each $k$ is
computationally prohibitive.
\section{Inference of binary independent components}
\label{sec:algo}
In this section, we first present a motivating example which provide the intuition
for our inference scheme, and then devise an efficient iterative procedure to
estimate $p_i$'s. The key challenge lies in that both $\G$ and $\MP(\yy)$ are
unknown.  If $\G$ is given, the problem becomes trivial and can be easily
solved by directly applying Maximum-likelihood type of methods.

\paragraph*{A motivating example}
\begin{figure}[thp]
\begin{center}

\includegraphics[width=1in]{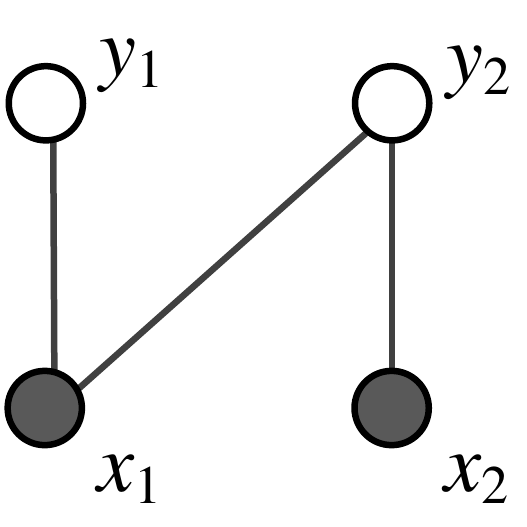}
$$
\G =
\begin{bmatrix}
1&1 \\
0&1 \\
\end{bmatrix}
$$
\caption{A simple mixture model. Hidden components are shown in
white disks and observable components are shown in black disks.}
\vspace{-1em}
\label{fig:toy}
\end{center}
\end{figure}

Consider a simple mixture model with 2 hidden sources and 2 observables as
depicted in Figure~\ref{fig:toy}.  The probability vector of the sources is $p
= [\mbox{0.2 0.5}]$ with component $y_1$ having the lower probability being
one. The marginal probabilities of $x_1 = 1$ and $x_2 = 1$ can be easily
computed as 0.6 and 0.5, respectively. Let the realizations of $y_1$ and $y_2$
over ten trials be:

\begin{center}
\vspace{-0.2in}
$$
\Y =
\begin{bmatrix}
0&0&1&0&0&0&0&0&1&0\\
0&1&1&0&0&1&1&1&0&0
\end{bmatrix},
$$
\end{center}
where $y_{it} = 1$ indicates that source $y_i = 1$ at the time slot $j$. $\Y$
is hidden and unknown. Since $\G = \begin{bmatrix} 1&1\\ 0&1 \end{bmatrix}$,
we have the observation matrix:

\begin{center}
\vspace{-0.2in}
$$
\X =
\begin{bmatrix}
0&1&1&0&0&1&1&1&1&0\\
0&1&1&0&0&1&1&1&0&0
\end{bmatrix}
$$
\end{center}

The objective of bICA  is to infer $\G$ and $p$ from $\X$. Since the number of
observables $m = 2$, we may identify at most $2^m -1 = 3$ unique sources, say,
$\hat{y}_1, \hat{y}_2, \hat{y}_3$. Denote the inferred $\G$ and $p$ as

\begin{center}
\vspace{-0.2in}
$$
\hat{\G} =
\begin{bmatrix}
1&0&1\\
0&1&1
\end{bmatrix},
\hat{p} =
\begin{bmatrix}
\hat{p}_1&\hat{p}_2&\hat{p}_3
\end{bmatrix},
$$
\end{center}
In another word, component $\hat{y}_1$ only manifests through $x_1$,
$\hat{y}_2$ through $x_2$, and $\hat{y}_3$ through both $x_1$ and $x_2$.  Clearly,
realizations in $\X$ where $x_2 = 0$ correspond to time slots when
sources $\hat{y}_2$ and $\hat{y}_3$ are both zero. In other words, we
have
$$
\MP(x_1 = 1, x_2 = 0) = \hat{p}_1(1 - \hat{p}_2)(1 - \hat{p}_3).
$$

Note that $\MP(x_2 = 0) = (1 - \hat{p}_2)(1 - \hat{p}_3)$. Therefore, we have
$\hat{p}_1 = \MP(x_1 = 1| x_2 = 0) \approx 0.2$. The last term is estimated
from the realization of $X$ where $x_2 = 0$.
Since $x_1 = 1$ if $\hat{y}_1 = 1$ or $\hat{y}_3 = 1$,
$\hat{p}_3$ can be calculated from
$$
1 - \MP(x_1 = 1) = (1 - \hat{p}_1)(1 - \hat{p}_3) \Rightarrow \hat{p}_3 = 0.5.
$$

Similarity, $\hat{p}_2$ can be calculated from
$$
1 - \MP(x_2 = 1) = (1 - \hat{p}_2)(1 - \hat{p}_3) \Rightarrow \hat{p}_2 = 0.
$$
$\hat{p}_2 = 0$ implies that $\hat{y}_2$ never activates, and thus its
associated column can be removed from  $\hat{\G}$.

The basic intuition of the above procedure is by limiting our considerations to
realizations where some observables are zero, we ``null" out the effects of
components that contribute to these observations.  This reduces the size of the
inference problem to be considered.

\paragraph*{The basic algorithm}
Motivated by the above example, we will consider $\G$ to be a $m \times 2^m$ adjacent matrix
for the complete bipartite graph. Furthermore, the columns of $\G$ are ordered
such that  $g_{kl} = 1$ if $l \wedge 2^k = 1$, for $k = 1, \ldots, m$, where $\wedge$
is the bit-wise AND operator. As an example, when $m=3$ we have:
$$
\G =
\begin{bmatrix}
0&1&0&1&0&1&0&1\\
0&0&1&1&0&0&1&1\\
0&0&0&0&1&1&1&1
\end{bmatrix}
$$
If the active probability of the $l$'th component $p_l = 0$, this implies the
corresponding column $\G_{:,l}$ can be removed from $\G$. Before proceeding to
the details of the algorithm, we first present a technical
lemma.
\begin{lem}
\label{lem:infer}
Consider a set  $\xx = [x_1, x_2, \ldots, x_{h-1}, x_h]^T$ generated by the data
model in \eqref{eq:boolean}, i.e., $\exists$ binary independent sources $\yy$,
s.t., $\xx = \G \otimes \yy$. The conditional random vector $\xx_{x_h = 0} = [x_1,
x_2, \ldots, x_{h-1}| x_h = 0]^T$ corresponds to the vector of the first $h-1$ elements of $\xx$ when $x_h =
0$. Then, $\xx_{x_h = 0} = \G '\otimes \yy'$, where $\G' = \G_{:,1 \ldots 2^{h-1}}$
(i.e. the first $2^{h-1}$ columns of $\G$)
and $\MP(y'_l = 1) =  \MP(y_l = 1)$ for $l = 1, \ldots, 2^{h-1}$.
\end{lem}
\begin{proof}
We first derive the conditional probability distribution of the first $h-1$ observation variables
given $x_h = 0$,
\begin{equation}
\begin{array}{ll}
& \MP(x_1, x_2, \ldots, x_{h-1} \mid x_{h} = 0)  \\
= & \MP(x_1, x_2, \ldots, x_{h-1} \mid x_{h} = 0)\MP(x_{h}=0) \\
\stackrel{(a)}{=} & \hspace{-3mm}\displaystyle \sum_{\yy \in Y(\xx)}{\prod_{l=1}^{2^h-1}{p_{l}^{y_{l}}(1-p_{l})^{1-y_{l}}}} \\
= & \hspace{-5mm}\displaystyle \sum_{\begin{array}{l} \yy_{1..2^{h-1}}\in Y(\xx_{1..h-1})\\ y_{l} = 0, \forall g_{hl}=1 \end{array}}{\prod_{g_{hl} = 0}{p_{l}^{y_{l}}(1-p_{l})^{1-y_{l}}}\prod_{g_{hl}=1}{(1-p_{l})}}
\end{array}
\end{equation}
(a) is due to (\ref{eq:relation}). Since $\MP(x_{h}=0) = \displaystyle \prod_{g_{hl}=1}{(1-p_l)}$, we have
\begin{equation}
\begin{array}{ll}
& \MP(x_1, x_2, \ldots, x_{h-1}\mid x_{h} = 0) \\
= & \displaystyle \sum_{\yy'\in Y(\xx_{1:h-1})}{\prod_{l=1}^{2^h-1}{(p'_{l})^{y'_{l}}(1-p'_{l})^{1-y'_{l}}}} \\
= & \displaystyle \sum_{\begin{array}{l} \yy_{1, \ldots, 2^{h-1}}\in Y(\xx_{1, \ldots, h-1})\\ y_{l} = 0, \forall g_{hl}=1 \end{array}}{\prod_{g_{hl} = 0}{p_{l}^{y_{l}}(1-p_{l})^{1-y_{l}}}.}
\end{array}
\end{equation}

Clearly, by setting $\MP(y'_l = 1) =  \MP(y_l = 1)$ for $l = 1, \ldots,
2^{h-1}$, the above equality holds.  In the other word, the conditional random
vector $\xx_{x_h = 0} = \G' \otimes \yy'$ for $\G' = \G_{:,1 \ldots 2^{h-1}}$.
\end{proof}

The above lemma establishes that the
conditional random vector $\xx_{x_h = 0}$ can be represented as an OR mixing of
$2^{h-1}$ independent components. Furthermore, the set of the independent
components is the same as the first $2^{h-1}$ independent components of $\xx$
(under proper ordering).

Consider a sub-matrix of data matrix $\X$, $\X_{(h-1)\times T}^0$,
where the rows  correspond to observations of $x_1, x_2, \ldots, x_{h-1}$
for $t = 1, 2, \ldots, T$ such that $x_{ht} = 0$. Define $\X_{(h-1)\times T}$, which
consists of the first $h-1$ rows of $\X$.  Suppose that we have computed the bICA
for data matrices $\X_{(h-1)\times T}^0$ and $\X_{(h-1)\times T}$. From
Lemma~\ref{lem:infer}, we know that $\X_{(h-1)\times T}^0$ is realization of OR mixtures of
independent components, denoted by $\yy_{2^{h-1}}^0$. Furthermore,
$\MP[\yy_{2^{h-1}}^0(l) = 1] =  \MP(y_l = 1)$ for $l = 1, \ldots, 2^{h-1}$.
Clearly, $\X_{(h-1)\times T}$ is realization of OR mixtures of $2^{h-1}$ independent
components, denoted by $\yy_{2^{h-1}}$.
Additionally, it is easy to see that the following holds:
\begin{equation}
\begin{array}{ll}
& \MP[\yy_{2^{h-1}}(l)=1]\\
& = 1 - [1-\MP(\yy_{2^{h-1}}^0(l) = 1)][1-\MP(y_{l+2^{h-1}} = 1)] \\
& = 1 - (1-p_l)(1-p_{l+2^{h-1}}),
\end{array}
\label{eq:half}
\end{equation}
where $l = 1, \ldots, 2^{h-1}$. Therefore,
\begin{equation}
\begin{array}{llll}
p_{l} & = & \MP(\yy^0_{2^{h-1}}(l)=1), & l = 1, \ldots, 2^{h-1}, \\
p_{l+2^{h-1}} & = & 1-\frac{1 - \MP(\yy_{2^{h-1}}(l)=1)}{1-\MP(\yy_{2^{h-1}}^0(l) = 1)}, & l = 2, \ldots, 2^{h-1}, \\
p_{2^{h-1} + 1} & = & \frac{\MF{(x_h=1 \wedge x_i=0, \forall i \in [1 \ldots h-1])}}
{\prod_{l=1 \ldots 2^h, l \neq 2^{h-1}-1}{(1-p_l)}}.
\end{array}
\label{eq:half2}
\end{equation}
The last equation above holds because realizations of $\xx$ where ($x_k=1$
while $x_i=0; \forall i \in \set{0,\ldots,k-1}$) are generated from OR mixtures of $\yy_{2^{k-1}}$'s only.

Let $\MF(A)$ be the frequency of event $A$,
we have the iterative inference algorithm as illustrated in Algorithm~\ref{algo:inc}.

\begin{algorithm}[h]
\caption{Incremental binary ICA inference algorithm}
\small
\label{algo:inc}
\SetKwData{Left}{left}
\SetKwInOut{Input}{input}
\SetKwInOut{Output}{output}
\SetKwInOut{Init}{init}
\SetKwFor{For}{for}{do}{endfor}
\SetKwFunction{FindBICA}{FindBICA}
\FindBICA($\X$)\\
\Input{Data matrix $\X_{m\times T}$}
\Init{$n = 2^{m} -1$\; $p_h = 0, h = 1, \ldots, n$\;
$\G$ = $m \times (2^{m} -1)$ matrix with rows corresponding all possible binary vectors of length $m$\;
$\varepsilon$ = the minimum threshold for $p_h$ to be considered a real component;} \BlankLine

\nl\eIf{$m = 1$} {\label{line1}
\nl    $p_1 = \MF(x_1 = 0)$\;\label{line2}
\nl    $p_2 = \MF(x_1 = 1)$\;\label{line3}
    } {
\nl    \eIf{$\X^0_{(m-1)\times T} = \emptyset$} {\label{line12}
\nl        $p_{1 \ldots 2^{m-1}}$ = \FindBICA($\X_{(m-1)\times T}$)\;\label{line13}
\nl        $p_{2^{m-1}+1} = 1$\;\label{line14}
\nl        $p_{2^{m-1}+2\ldots 2^m} = 0$\;\label{line15}
        } {
\nl    $p_{1 \ldots 2^{m-1}}$ = \FindBICA($\X^0_{(m-1)\times T}$)\;\label{line4}
\nl    $p'_{1 \ldots 2^{m-1}}$ = \FindBICA($\X_{(m-1)\times T}$)\;\label{line5}
\nl    \For{$l = 2, \ldots, 2^{m-1}$}{\label{line6}
\nl        $p_{l+2^{m-1}} = 1-\frac{1-p'_l}{1-p_l}$\;\label{line7}
       }
\nl    $p_{2^{m-1}+1} = \frac{\MF{(x_m=1 \wedge x_i=0, \forall i \in [1 \ldots m-1])}}{\prod_{l=1 \ldots 2^m-1, l \neq 2^{m-1}+1}{(1-p_l)}}$\;\label{line8}
}
}
\nl\For{$h = 1, \ldots, 2^{m}$}{\label{line9}
\nl    \If{$(p_h < \varepsilon) \vee (p_h = 0)$}{\label{line10}
\nl            prune $p_h$ and corresponding columns $g_{h}$\;\label{line11}
        }
}
\nl\Output{$p$ and $\G$}
\end{algorithm}
\normalsize

When the number of observation variables $m=1$, there are only two possible unique sources,
one that can be detected by the monitor $x_1$, denoted by [1]; and one that cannot,
denoted by [0]. Their active probabilities can easily be calculated by counting
the frequency of $(x_1 = 1)$ and $(x_1 = 0)$ (lines~\ref{line1} --
\ref{line3}). If $m \ge 2$, we apply Lemma~\ref{lem:infer} and \eqref{eq:half2}
to estimate $p$ and $\G$ through a recursive process.
$\X^0_{(m-1)\times T}$ is sampled from columns of $\X$ that have $x_m = 0$. If $\X^0_{(m-1)\times T}$
is an empty set (which means $x_{mt} = 1, \forall t$) then we can associate $x_m$
with a constantly active component and set the other components' probability accordingly
(lines~\ref{line12} -- \ref{line15}). If $\X^0_{(m-1)\times T}$ is non-empty, we invoke {\sc FindBICA} on
two sub-matrices $\X^0_{(m-1)\times T}$ and $\X_{(m-1)\times T}$ to determine
$p_{1 \ldots 2^{m-1}}$ and $p'_{1 \ldots 2^{m-1}}$, then
infer $p_{2^{m-1}+1 \ldots 2^{m}}$ as in (\ref{eq:half2}) (lines~\ref{line6} --
\ref{line8}). Finally, $p_h$ and its corresponding column $g_h$ in $\G$ are
pruned in the final result if $p_h < \varepsilon$ (lines~\ref{line9} --
\ref{line11}).

\paragraph*{Reducing computation complexity}
Let $S(m)$ be the computation time for finding bICA given $\X_{m\times T}$. From
Algorithm~\ref{algo:inc}, we have, $$S(m) = 2S(m-1) + c2^m,$$ where $c$ is a constant. It is
easy to verify that $S(m) = cm 2^m$.  Therefore, Algorithm~\ref{algo:inc} has an
exponential computation complexity with respect to $m$. This is clearly
undesirable for large $m$'s. However, we notice that in practice, correlations
among $x_i$'s exhibit locality, and the $\G$ matrix tends to be sparse. Instead
of using a complete bipartite graph to represent $\G$, the degree of vertices in
$V$ (or the number of non-zero elements in $\G_{:,k}$) tend to be much less than
$m$. In what follows, we discuss a few techniques to reduce the computation complexity by discovering
and taking advantage of the sparsity of $\G$. We first establish a few technical
lemmas.

\begin{lem}
If $x_i$ and $x_k$ are uncorrelated, then $\MP(y_l = 1) = 0$, $\forall l$ s.t.,
$g_{il} = 1$ and $g_{kl} = 1$.
\label{lem:5}
\end{lem}
\begin{proof}
We prove by contradiction. Suppose $\exists l$, s.t., $g_{il} = 1$, $g_{kl} =
1$, and $\MP(y_l = 1) = 0$. Denote the $l$'s by a set $L$. Let $u = \wedge_{l \in L}y_l$.
From the assumption, $u$ is non-degenerate.  Without loss of generality, we can
represent $x_i$ and $x_k$ as
\begin{eqnarray}
\nonumber
x_i = u \vee v_1 \\
\nonumber
x_k = u \vee v_2
\end{eqnarray}
where $v_1$ and $v_2$ are disjunctions of remaining non-overlapping components in $x_i$ and $x_k$, respectively
From Lemma~\ref{lem:corr}, we know that $x_i$ and $x_k$ are correlated. This
contradicts the condition.
\end{proof}

\begin{lem}
Consider the conditional random vector $\xx_{x_k = 0} = [x_1, x_2, \ldots,
x_{k-1}| x_k = 0]^T$ from a set  $\xx = [x_1, x_2, \ldots, x_{k-1}, x_k]^T$
generated by the data model in \eqref{eq:boolean}. If $x_i$ and $x_k$ are
uncorrelated, $\xx_{x_k = 0}(i)$ and $\xx_{x_k = 0}(k)$ are uncorrelated.
\label{lem:6}
\end{lem}
\begin{proof}
This lemma is a direct consequence of Lemma~\ref{lem:infer} and Lemma~\ref{lem:5}.
\end{proof}

Lemma~\ref{lem:5} implies that pair-wise independence can be used to eliminate
edges/columns in $\G$. Lemma~\ref{lem:6} states that the pair-wise independence
remains true for conditional vectors. Therefore, we can treat the conditional
vectors similarly as the original ones.

We also observe that for $\xx = [x_1, x_2, \ldots, x_{k-1}, x_k]^T$ if
\eqref{eq:cond} holds then there does not exist an independent component that
generates $x_k$ and some of $x_j$, $j = 1, 2, \ldots k-1$. In other words,
$x_k$ is generated by a ``separate" independent component.
\begin{equation}
\MP(x_1, x_2, \ldots, x_{k-1}, x_k) = \MP(x_1, x_2, \ldots, x_{k-1})P(x_k)
\label{eq:cond}
\end{equation}
Finally from \eqref{eq:half}, we see that $\MP(y_{k-1}(l) = 1) \ge \max(p_l,
p_{l+2^{k-1}})$. Note that $\MP(y_{k-1}(l) = 1)$ is inferred from
$\X_{(k-1)\times T}$, while the latter two are for $\X_{k\times T}$.  This
property allows us to prune the $\G$ matrix along with the iterative procedure
(as opposed to at the very end).

Now we are in the position to outline our complexity reduction techniques.
\begin{itemize}
\item[T1] For every pair $i$ and $k$, compute $$Cov(i,k) =
\frac{\sum_{t}^T{\X_{it}\X_{kt}}}{T} -
\frac{\sum_{t}^T{\X_{it}}}{T}\frac{\sum_{t}^T{\X_{kt}}}{T}.$$
Let the associated $p$-value be $p(i,k)$. The basic idea of $p$-value is to use
the original paired data $(\X_i, \Y_i)$, randomly redefining the pairs to create a
new data set and compute the associated $r$-values. The p-value for the
permutation is proportion of the $r$-values generated that are larger than that
from the original data.  If $p(i,k) > \epsilon$, where $\epsilon$ is a small
value (e.g., 0.05), we can remove the corresponding columns in $\G$ and
elements in $\yy$.
\item[T2] We can determine the bICA for each sub-vector separately if the following holds,
$$\MP(x_1, \ldots, x_k) = \MP(x_1, \ldots, x_{l})\MP(x_{1+1}, \ldots, x_{k}).$$
\item[T3] If the probability of the $i$'th component of $\X_{k\times T}$ $p_i  <
\epsilon$, then $\forall j$, s.t., $\G_{:,i} \prec \G_{:,j}$, the probability of
the $j$'th component of $\X_{k'\times T}$ $p_j < \epsilon$ for $k' > k$. In
another word, these columns and corresponding components can be eliminated.
\end{itemize}
From our evaluation study, we find the computation time is on the order of seconds for a problem size $m = 20$ on a regular desktop PC.
\section{The Inverse Problem}
\label{sec:inverse}
Now we have the mixing matrix $\G_{m\times n}$ and the active probabilities
$\MP(\yy)$, given observation $\X_{m\times T}$, the inverse problem concerns
inferring the realizations of the latent variables $\Y_{n\times T}$.
%
Recall that $n$ is the number of latent variables.
Denote $y_i$ to be the binary variable for the
$i$'th latent variable. Let $\xx = \G \otimes \yy$.  We assume that the probability of
observing $\X$ given $\xx$ depends on their Hamming distance $d(\xx, \X) =
\sum_{i}{|\X_i-x_i|}$, and $\MP(\xx|\X) = p_e^{d(\xx, \X)}(1-p_e)^{m-d(\xx, \X)}$, where $p_e$ is the
error probability of the binary symmetric channel. To determine $\yy$, we
can maximize the posterior probability of $\yy$ given $\X$ derived as follows,
$$
\begin{array}{lll}
\MP\{\yy|\X\} & = & \frac{\MP\{\X|\yy\}\MP\{\yy\}}{\MP\{\X\}} \\
         & = & \frac{\MP\{\X|\yy\}\MP\{\yy\}}{\MP\{\X\}} \\
         & \stackrel{(a)}{=} & \frac{\MP\{\X,\xx|\yy\}\MP\{\yy\}}{\MP\{\X\}} \\
         & \stackrel{(b)}{=}& \frac{\MP\{\X|\xx\}\MP\{\yy\}}{\MP\{\X\}} \\
         & = & \frac{\prod_{i=1}^{m}{\MP\{\X_i|x_i\}}\prod_{j=1}^{n}{\MP\{y_i\}}}{\MP\{\X\}} \\
         & = & \frac{\prod_{i=1}^{m}{ p_e^{|x_i- \X|}(1-p_e)^{1-|x_i-\X|}}\prod_{j=1}^{n}{p_i^{y_i}(1-p_i)^{1-y_i}}}{\MP\{\X\}}
\end{array}
$$
where $\xx = \G \otimes \yy$.  $(a)$ and $(b)$ are due to the deterministic
relationship between $\xx$ and $\yy$.  
{Recall that $x_i = \bigvee_{j=1}^{n}{(g_{ij}\wedge y_j)}, i = 1, \ldots, m$.
With $M$ is a ``large enough'' constant, we can use big-$M$ formulation}~\cite{Griva2008Linear}
{to relax the disjunctive set and convert the above relationship between $\xx$ and
$\yy$ into the following two sets of conditions:}
\begin{equation}
\begin{array}{llll}
 \displaystyle x_i & \le & \sum_{j=1}^{n}{g_{ij}y_{j}}, \mbox{ } \forall i=1,\ldots,m. \\
 \displaystyle M\cdot x_i & \ge & \sum_{j=1}^{n}{g_{ij}y_{j}}, \mbox{ } \forall i=1,\ldots,m. \\
\end{array}
\end{equation}
Here, since $\sum_{j=1}^{n}{g_{ij}y_{j}} \le n$, we can set $M = n$.
Finally, taking $\log$ on both sides  and introducing additional auxiliary variable
$\alpha_i = |\X_i - x_i|$, we {have} the the following integer programming
problem:
\begin{equation}
\begin{array}{lll}
\underset{\alpha,y}{\max.}& \displaystyle \sum_{i=1}^m{\left[\alpha_i\log{p_e} + (1-\alpha_i)\log(1-p_e)\right]}\\
    & + \sum_{j=1}^n{\left[(1-y_j)\log{(1-p_j)} + y_j\log{p_j}\right]} \\
\mbox{s.t.}& \displaystyle x_i \le \sum_{j=1}^{n}{g_{ij}y_{j}}, \hspace{10mm} \forall i=1,\ldots,m, \\
    & \displaystyle n\cdot x_i \ge \sum_{j=1}^{n}{g_{ij}y_{j}}, \hspace{5.5mm} \forall i=1,\ldots,m, \\
    & \alpha_i \ge \X_i - x_i, \hspace{11mm} \forall i=1,\ldots,m, \\
    & \alpha_i \ge x_i - \X_i, \hspace{11mm} \forall i=1,\ldots,m, \\
    & \alpha_i, x_i, y_j = \set{0,1}, \hspace{4.5mm} \forall i=1,\ldots,m, j=1,\ldots, n. \\
\end{array}
\label{eq:IPP}
\end{equation}
This optimization function can be solved using ILP solvers. Note that $p_e$ can
be thought of the penalty for mismatches between $x_i$ and $\X_i$.
\normalsize
\paragraph*{Zero Error Case}
If $\X$ is perfectly observed, containing no noise,
we have $p_e = 0$ and $\alpha_i = \xx_i - \X_i = 0$, or equivalently,
$\xx_i = \X_i$. The integer programming problem in \eqref{eq:IPP}
can now be simplified as:
\begin{equation}
\begin{array}{lll}
\underset{y}{\max.}& \displaystyle \sum_{j=1}^n{\left[(1-y_j)\log{(1-p_j)} + y_j\log{p_j}\right]} \\
\mbox{s.t.}& \displaystyle \X_i \le \sum_{j=1}^{n}{g_{ij}y_{j}}, \hspace{10mm} \forall i=1,\ldots,m, \\
    & \displaystyle n\cdot \X_i \ge \sum_{j=1}^{n}{g_{ij}y_{j}}, \hspace{5mm} \forall i=1,\ldots,m, \\
    & y_j = \set{0,1}, \hspace{15mm} \forall j=1,\ldots, n. \\
\end{array}
\label{eq:IPPSim}
\end{equation}

Clearly, the computation complexity of the zero error case is lower
compared to \eqref{eq:IPP}. It can also be used in the case where
prior knowledge regarding the noise level is not available.

\section{Evaluation}
\label{sec:eval}
In this section, we first introduce performance metrics used for the evaluation
of the proposed method, and then present its performance under different
network topologies with different level of observation noises. In evaluating
the structural errors of bICA, we also make an independent contribution by
devising a matching algorithm for two bipartite graphs.

We compare the
proposed algorithm with the Multi-Assignment Clustering (MAC)
algorithm~\cite{Streich09}. The source code of MAC was obtained from the
authors' web site.  As shown in Table~\ref{tab:comparison}, none of existing algorithms
follows exactly the same model and/or constraints as bICA.  For instance, MAC
assumes the knowledge of the dimension of latent variables. Therefore, the
comparisons bias against our proposed algorithm.

%
%
%
\subsection{Evaluation metrics}
We denote by $\hat{p}$ and $\hat{\G}$ the inferred active probability of
latent variables and the inferred mixing matrix, respectively.
%
%
\subsubsection{In Degree Error \& Structure Error}
Let $\|\cdot\|_1$ be the 1-norm defined by $\|x\|_1 =
\sum_{i=1}^{n}\sum_{j=1}^{m}{|x_{ij}|}$. $diag(\cdot)$ and
$triu(\cdot)$ denote the main diagonal and the
upper triangular portion of a matrix, respectively.
Two metrics are introduced in \cite{computer06anon-parametric} to evaluate the dissimilarity of $\G$ and $\hat{\G}$.
{\em In degree error} $E_d$ is the difference between the true
in-degree of $\G$ and the expected in-degree $\hat{\G}$. It is
computed by taking the sum absolute difference between
$diag(\G\G^T)$ and $diag(\hat{\G}\hat{\G}^T)$ as the following:
\beq
\begin{array}{lll}
E_d & \stackrel{\Delta}{=} & \|diag(\G\G^T) -
diag(\hat{\G}\hat{\G}^T)\|_1
\end{array}
\eeq
The second metric,
{\em structure error} $E_s$ is the sum of absolute
difference between the upper triangular portion of $\G\G^T$ and
$\hat{\G}\hat{\G}^T$, defined as:
\beq
\begin{array}{lll}
E_s & \stackrel{\Delta}{=} & \|triu(\G\G^T) -
triu(\hat{\G}\hat{\G}^T)\|_1
\end{array}
\eeq
Since each element of the upper triangular
portion of $\G\G^T$ is a count of the number of hidden causes
shared by a pair of observable variables, the sum difference is
a general measure of graph dissimilarity.
\subsubsection{Normalized Hamming Distance}
This metric indicates how accurate the mixing matrix is estimated. It is defined by the Hamming
distance between $\G$ and $\hat{\G}$ divided by its size.
\beq
\begin{array}{lll}
\bar{H}_g & \stackrel{\Delta}{=} & \frac{1}{mn} \sum_{i=1}^{n} d^H(\G_{:,i},\hat{\G}_{:,i}).
\end{array}
\eeq
To estimate $\bar{H}_g$ however, two challenges remain:
First, the number of inferred independent
components may not be identical as the ground truth. Second, the order of
independent components in $\G$ and $\hat{\G}$ may be different.

To solve the first problem, we can either prune $\hat{\G}$ or introduce columns into $\G$
to equalize the number of components ($n = \hat{n}$, where $\hat{n}$ is the number of columns
in $\hat{\G}$). For the second problem, we propose a matching algorithm that minimizes the
Hamming distance between $\G$ and $\hat{\G}$ by permuting the {corresponding} columns in $\hat{\G}$.

\paragraph*{Structure Matching Problem}
A naive matching algorithm needs to consider all $\hat{n}!$
column permutations of $\hat{\G}$, and chooses the one that has the minimal Hamming
distance to $\G$. This approach incurs an exponential computation complexity. Next,
we first formulate the best match as an Integer Linear Programming problem. Denote the
Hamming distance between column $\hat{\G}_{:,i}$ and $\G_{:,j}$ as
$c_{ij} \geq 0$. Define a permutation matrix $\A_{n \times n}$ with $a_{ij} = 1$
indicating that the $i$'th column in $\hat{\G}$ is matched with the $j$'th column in $\G$.
The problem now is to find a permutation matrix such that the total
Hamming distance between $\G$ and $\hat{\G}$ (denoted by $d^H(\G,\hat{\G})$)
is minimized. We can formulate this problem as an ILP as follows:
\begin{equation}
\begin{array}{ll}
\underset{a}{\min.}& \displaystyle \sum_{i=1}^n{\sum_{j=1}^n{c_{ij} a_{ij}}} \\
\mbox{s.t.}& \displaystyle \sum_{i=1}^n{a_{ij}} = 1, \\
    & \displaystyle \sum_{j=1}^n{a_{ij}} = 1, \\
    & a_{ij} = \set{0,1}, \hspace{5mm} \forall i,j=1, \ldots ,n.
\end{array}
\label{eq:pgmatching}
\end{equation}
\begin{algorithm}[tp]
\caption{Bipartite graph matching algorithm} \small
\label{algo:pg} \SetKwData{Left}{left} \SetKwInOut{Input}{input}
\SetKwInOut{Output}{output} \SetKwInOut{Init}{init}
\SetKwFor{For}{for}{do}{endfor} \SetKwFunction{Hungarian}{BipartiteMatching}
\SetKwFunction{MatchPG}{MatchPG}
\MatchPG($\G$, $\hat{\G}$, $\hat{p}$)\\
\Input{$\G_{m \times n}$, $\hat{\G}_{m \times \hat{n}}$, $\hat{p}_{1 \times \hat{n}}$; ($n \le \hat{n} \le 2^{m}$)}
\Init{$\hat{\G}'_{m \times \hat{n}} = 0$; $\hat{p}'_{1 \times \hat{n}} = 0$; $\C_{\hat{n} \times \hat{n}} = 0$\;}
\BlankLine
\nl\For{$i = 1, \ldots, \hat{n}$}{\label{Line1}
\nl    \For{$j = 1, \ldots, \hat{n}$}{\label{Line2}
\nl        \eIf{$g_{i} = 0$} {\label{Line3}
\nl            $c_{ij} = d^H(\G_{:,i}, \hat{\G}_{:,j}) \times m$\;\label{Line4}
        } {
\nl            $c_{ij} = d^H(\G_{:,i}, \hat{\G}_{:,j})$\;\label{Line5}
        }
    }
}
\nl$\A$ = \Hungarian($\C$)\;\label{Line6}
\nl\For{$i = 1, \ldots, \hat{n}$}{\label{Line7}
\nl    find $j$ such that $a_{ij} = 1$\;\label{Line8}
\nl    $\hat{\G}'_{:,i} = \hat{\G}_{:,j}$\;\label{Line9}
\nl    $\hat{p}'_i = \hat{p}_j$;\label{Line10}
}
\nl Prune $\hat{\G}'$: $\hat{\G}' = \hat{\G}'_{:,1 \ldots n}$\;\label{Line11}
\nl Prune $\hat{p}'$: $\hat{p}' = \hat{p}'_{1 \ldots n}$\;\label{Line12}
\nl\Output{$\hat{\G}'$ and $\hat{p}'$}\label{Line13}
\end{algorithm}
\normalsize

The constraints ensure the resulting $\A$ is a permutation matrix. This problem
 can be solved using ILP solvers. However, we observe that the ILP is equivalent to
a maximum-weight bipartite matching problem. In the bipartite graph, the vertices
are positions of the columns, and the edge weights are the Hamming distance of the
respective columns. If we consider $d^H(\G_{:,i},\hat{\G}_{:,i})$, the Hamming distance between column $\G_{:,i}$
and $\hat{\G}_{:,i}$, to be the ``cost'' of matching $\hat{\G}_{:,i}$
to $\G_{:,i}$, then the maximum-weight bipartite matching problem can be solved
in $O(n^3)$ running time~\cite{KuhnHungarian}, where $n$ is
the number of vertices. The algorithm requires $\G$ and $\hat{\G}$ to have the same
number of columns.

One {greedy} solution is to prune $\hat{\G}$ by selecting the top $n$ components from
$\hat{\G}$, which have the highest associated probabilities $\hat{p}_{i}$ since
they are the most likely true components. However, when $T$ is small and/or under
large noise, we may not have sufficient observations to correctly identify
components in $\G$ with high confidence. As a result, true components might
have lower active probabilities comparing to the noise components. To address
the problem, we instead keep a larger $\hat{n}$ and introduce $\hat{n} - n$
artificial components into $\G$. These components will be represented by zero
columns in $\G$.
While matching the inferred columns in $\hat{\G}$ to the columns in $\G$, clearly an
undesirable scenario occurs when we accidently match a column in $\hat{\G}$ to an
additive zero column in $\G$. This happens when an inferred column $\hat{\G}_{:,i}$ is sparse
(i.e. having a very small Hamming distance to the zero column). To avoid the
incident, we multiply the cost of matching any column in $\hat{\G}$ to
a zero column in $\G$ by $m$. This eliminates the case {in which} a column $\hat{\G}_{:,i}$
is matched with a zero column in $\G$, since it is more expensive than matching
with another non-zero column $\G_{:,i}$.
We can now select the best $n$ candidates in $\hat{\G}$, which yields
a reduced mixing matrix $\hat{\G}'$ of size $m \times n$, and elements in active
probability vector $\hat{p}'$ will also be selected accordingly. The solution
to the structure matching problem is detailed in Algorithm \ref{algo:pg}.

\begin{figure}[tp]
\begin{center}
\hspace{-0.4in}\includegraphics[width=3.2in]{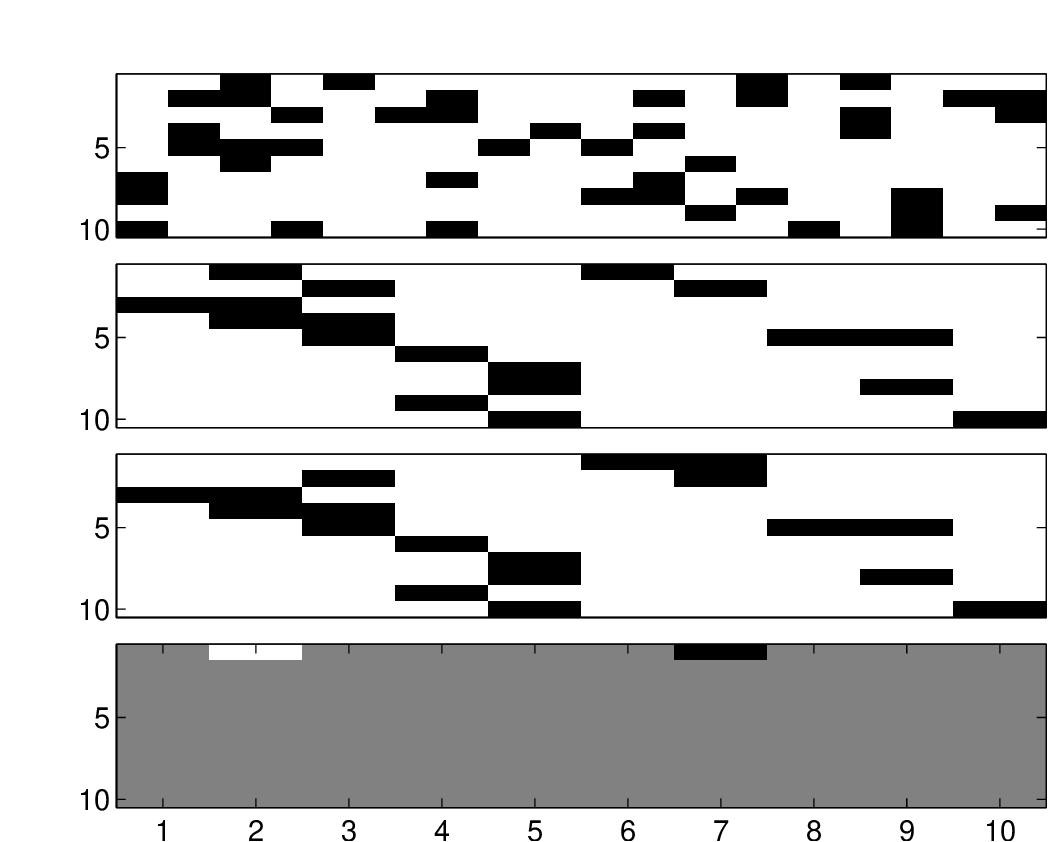}
\end{center}
\caption{From top to bottom: inferred matrix $\hat{\G}$ with 18
inferred components, transformed matrix $\hat{\G}'$ with only 10
components remaining (8 noisy ones were removed), original $\G$
and difference matrix between $\G$ and $\hat{\G}'$. Black dot = 1
and white dot = 0.} \label{fig:illus}
\end{figure}

In the algorithm,  lines~\ref{Line1} -- \ref{Line5} build the input weight matrix
$\C_{\hat{n} \times \hat{n}}$ for the bipartite matching algorithm. If $\G_{:,i}$ is
a zero column, $c_{ij}$ will be scaled by $m$ to avoid the matching
between column $\G_{:,i}$ and $\hat{\G}_{:,j}$ (line~\ref{Line4}). The bipartite matching algorithm
finds the optimal permutation matrix $\A$
to transform $\hat{\G}$ into $\hat{\G}'$ that is ``closest'' to $\G$ (lines
\ref{Line6} -- \ref{Line10}). We are only interested in the first $n$ columns
of $\hat{\G}'$ and $\hat{p}'$ (as they most likely represent the true PUs).
Therefore, $\hat{\G}'$ and $\hat{p}'$ are pruned in lines~\ref{Line11} -- \ref{Line13}.

As an example, the inferred result of a random network with $n = m = 10$ is given in Figure~\ref{fig:illus}.
Non-zero entries and zero entries of $\G$,
$\hat{\G}$, and $\hat{\G}'$ are shown as black and white dots,
respectively. The entry-wise difference matrix $|\G-\hat{\G}'|$ is
given in the bottom graph. {Gray dots in the difference matrix
indicate identical entries in the inferred $\hat{\G}$ and the original
$\G$}; and black dots indicate different entries (and thus errors in the inferred matrix).
In this case, only the first row
(corresponding to the first monitor $x_1$) contains some errors.
\subsubsection{Transmission Probability Error}
The prediction error in the inferred transmission probability of
independent users is measured by the Kullback-Leibler divergence
between two probability distributions $p$ and $\hat{p}$. Let $p'$ and $\hat{p}'$
denotes the ``normalized'' $p$ and $\hat{p}$ ($p'_i = p_i\sum_{i=1}^{n}p_i$),
Transmission Probability Error is defined as below (the K-L distance):
\beq
\begin{array}{lll}
\bar{P}(p',\hat{p}') & \stackrel{\Delta}{=} & \sum_{i=1}^{n}p_i\log(\frac{p_i'}{\hat{p_i}'}).
\end{array}
\eeq
Intuitively, Transmission Probability Error gets larger as the predicted
probability distribution $\hat{p}$ deviates more from the real distribution
$p$.

\subsubsection{Activity Error Ratio}
After applying {\sc FindBICA} in Algorithm~\ref{algo:inc} on the measurement data of length $T$ to obtain
$\hat{\G}$ and $\hat{p}$, realizations of the hidden variables
can be computed by solving the maximum likelihood estimation problem in (\ref{eq:IPP}). We define
\beq
\begin{array}{lll}
\bar{H}_y & \stackrel{\Delta}{=} & \frac{1}{nT} \sum_{i=1}^{T} d^H(\Y_{:,i},\hat{\Y}_{:,i}),
\end{array}
\eeq
where $\Y_{:,i}$ is the $i$'th column of $\Y$.
Similar to $\bar{H}_g$, this metric measures how accurately the activity
matrix is inferred by calculating the ratio between the size of $\yy$ and the absolute
difference between $\yy$ and $\hat{\yy}$.
\subsection{Experiment Results}
\begin{figure*}[tp]
\begin{center}
\begin{tabular}{ m{0.1in} m{1.5in} m{1.5in} m{1.5in} m{1.5in} }
(a) &
\hspace{0.2in}\includegraphics[width=1.1in]{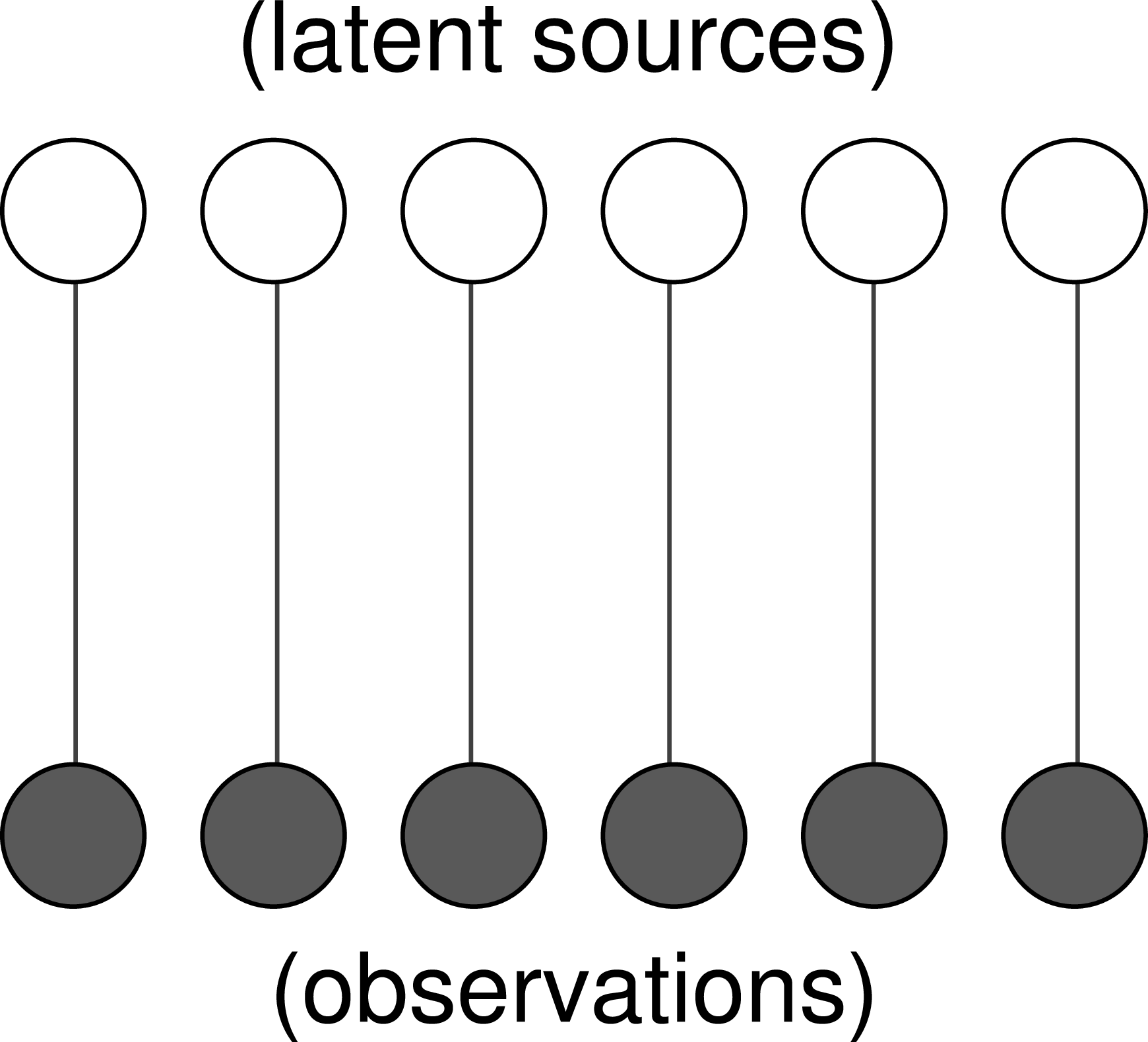}\hspace{0.2in} &
\hspace{0.2in}\includegraphics[width=1.1in]{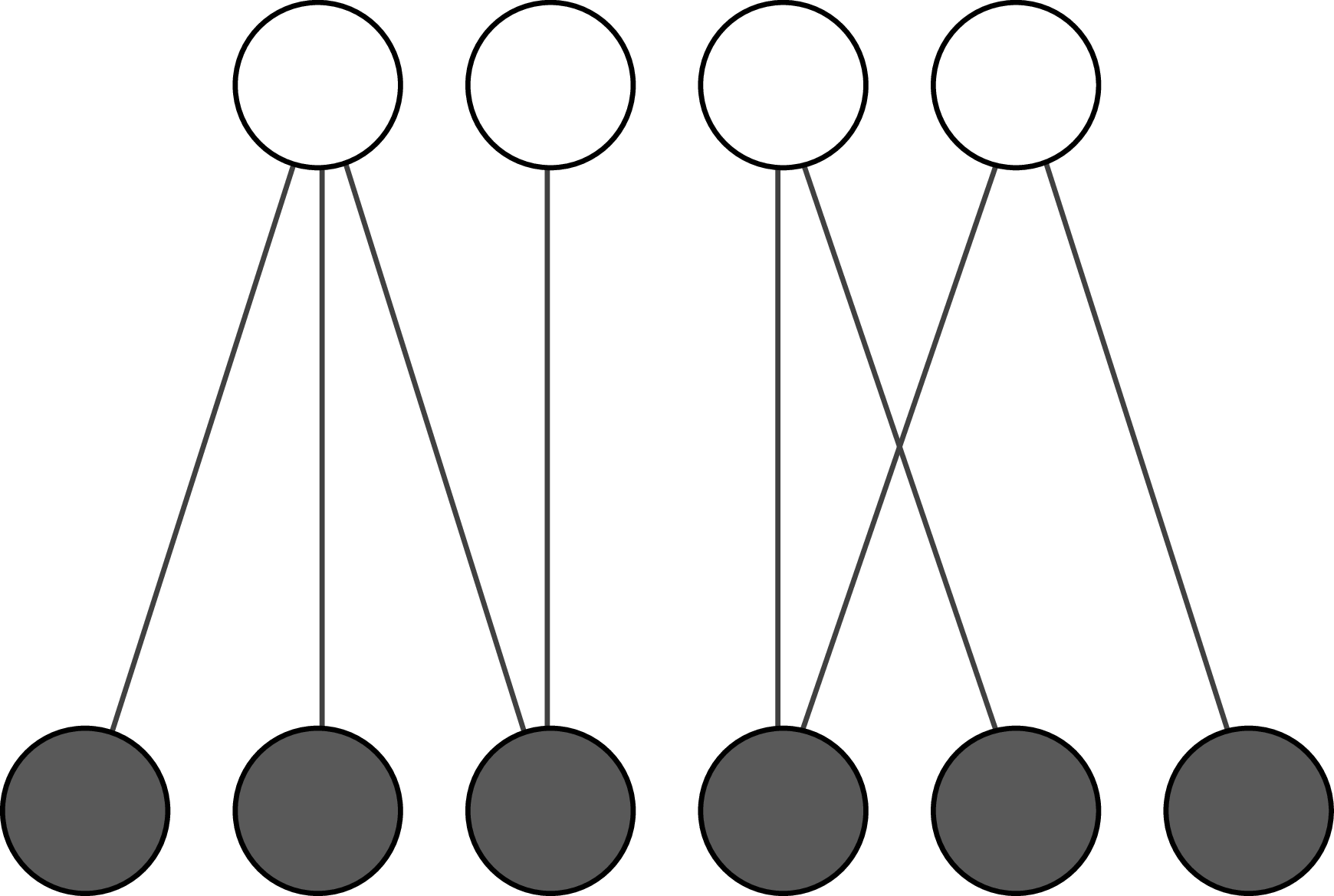}\hspace{0.2in} &
\hspace{0.2in}\includegraphics[width=1.1in]{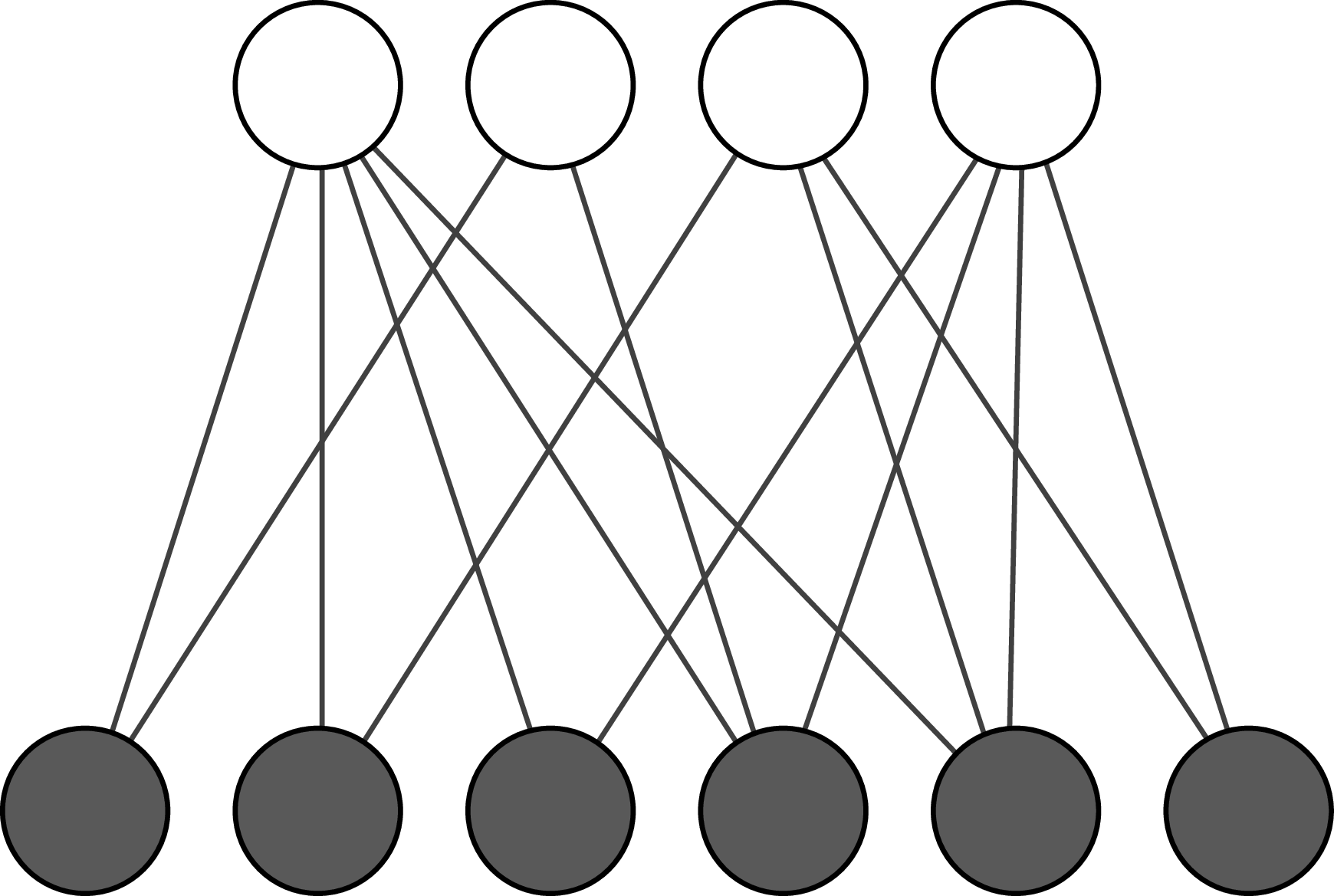}\hspace{0.2in} &
\includegraphics[width=1.5in]{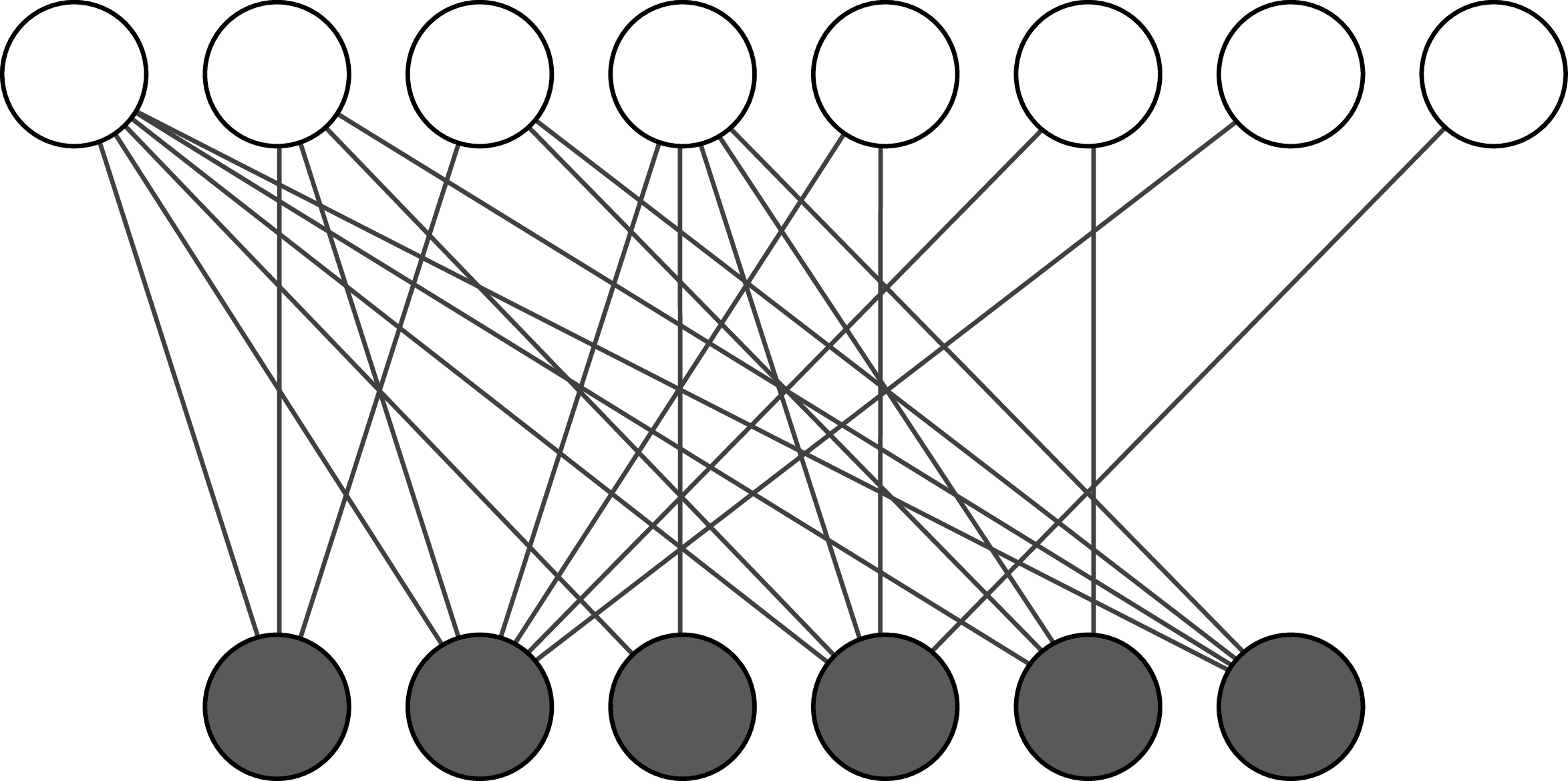}\\
(b) &
\includegraphics[width=1.5in]{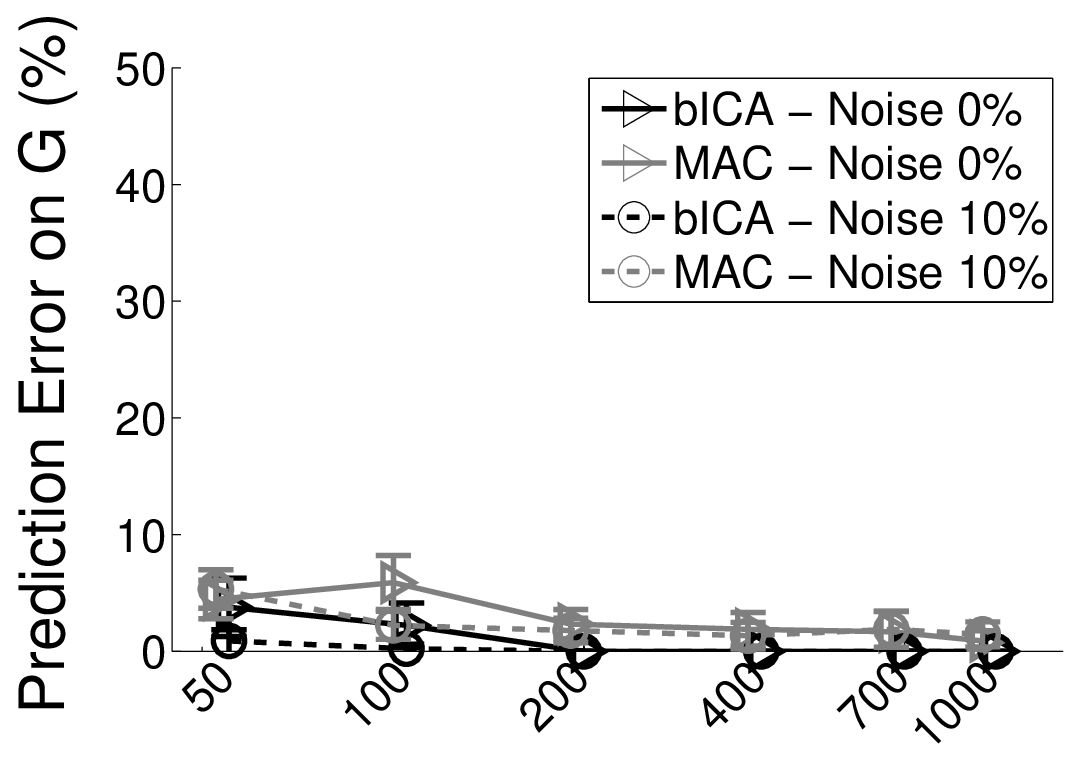} &
\includegraphics[width=1.5in]{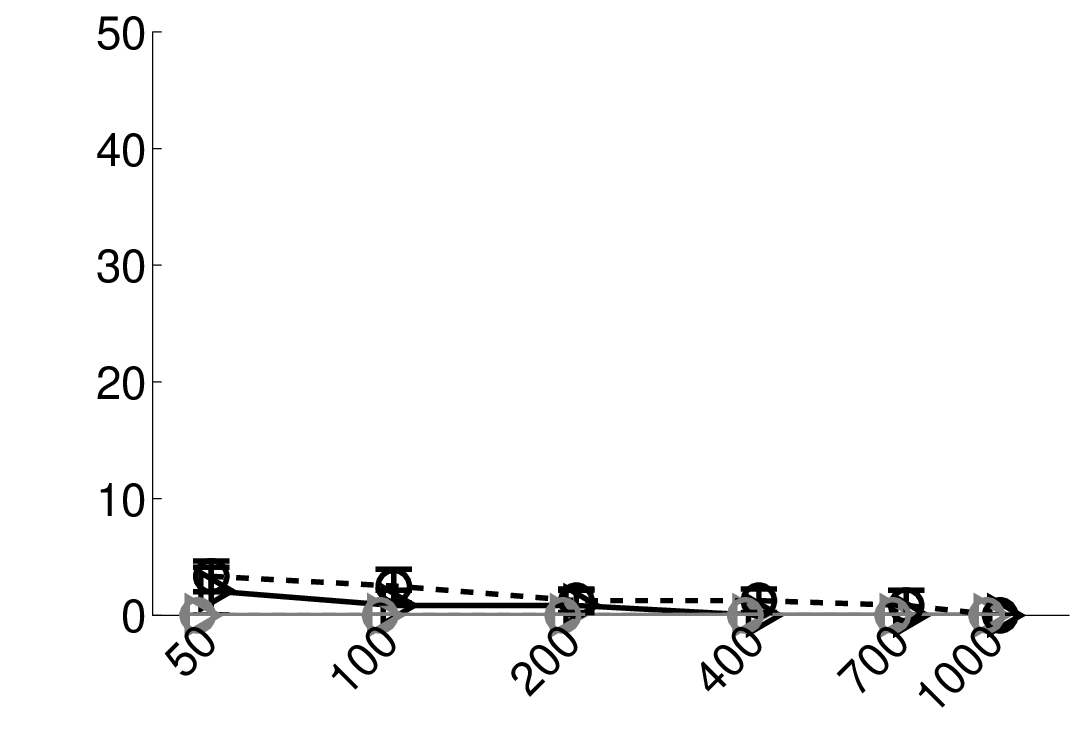} &
\includegraphics[width=1.5in]{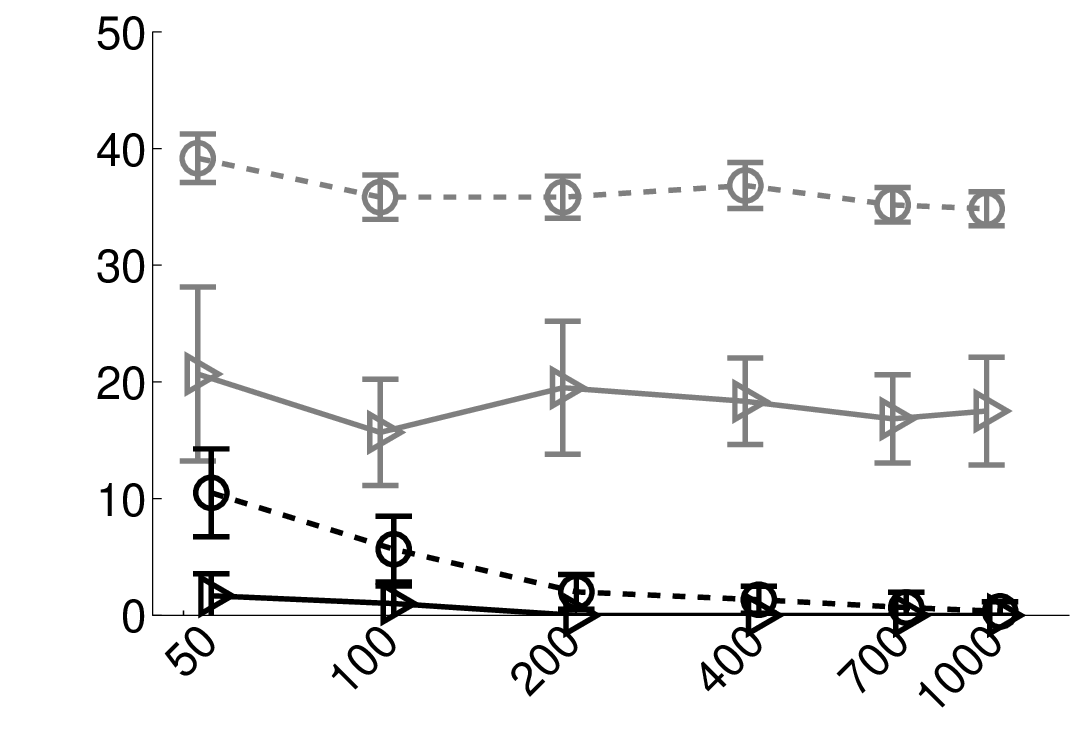} &
\includegraphics[width=1.5in]{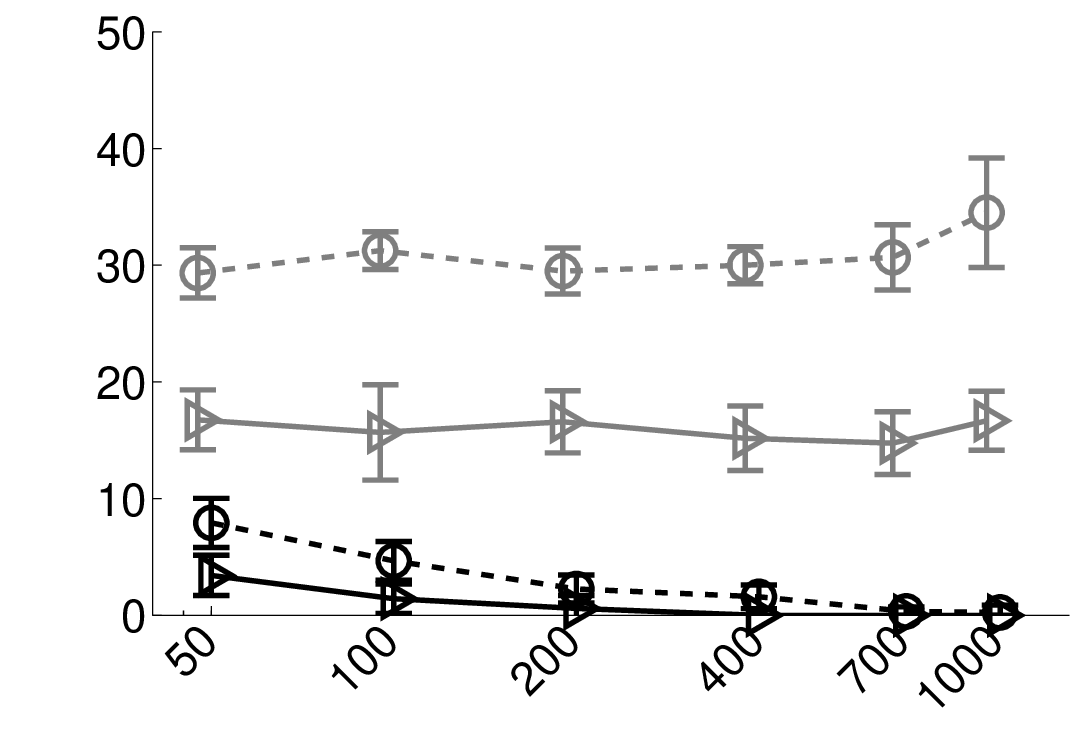}\\
(c) &
\includegraphics[width=1.5in]{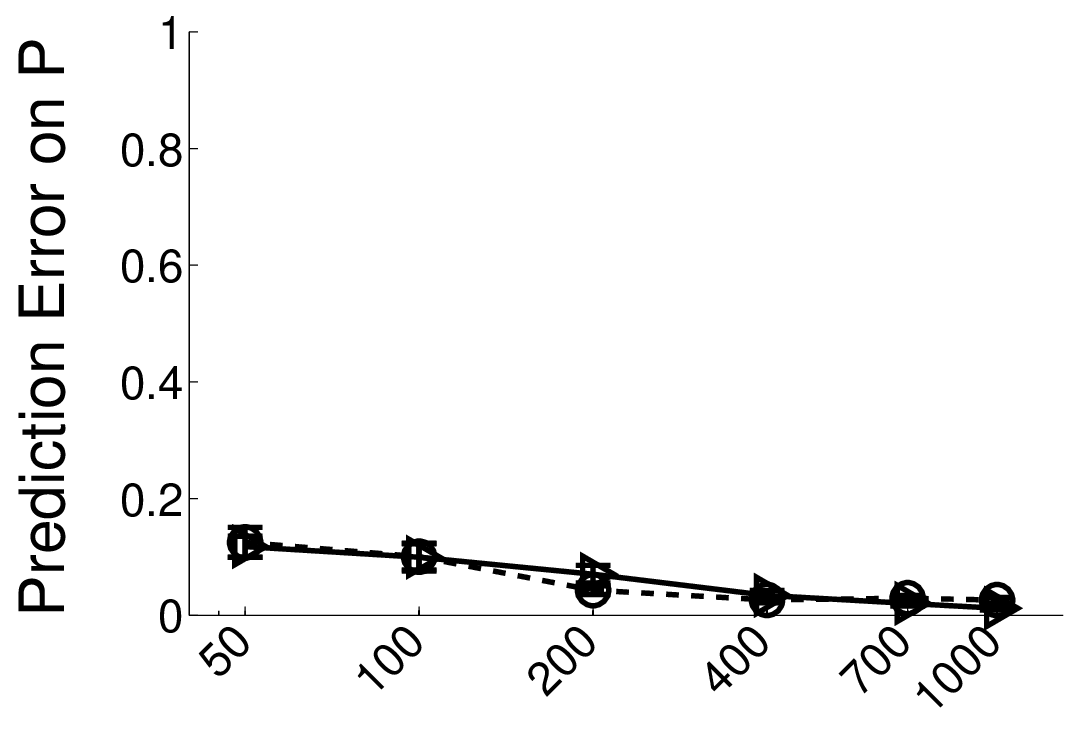} &
\includegraphics[width=1.5in]{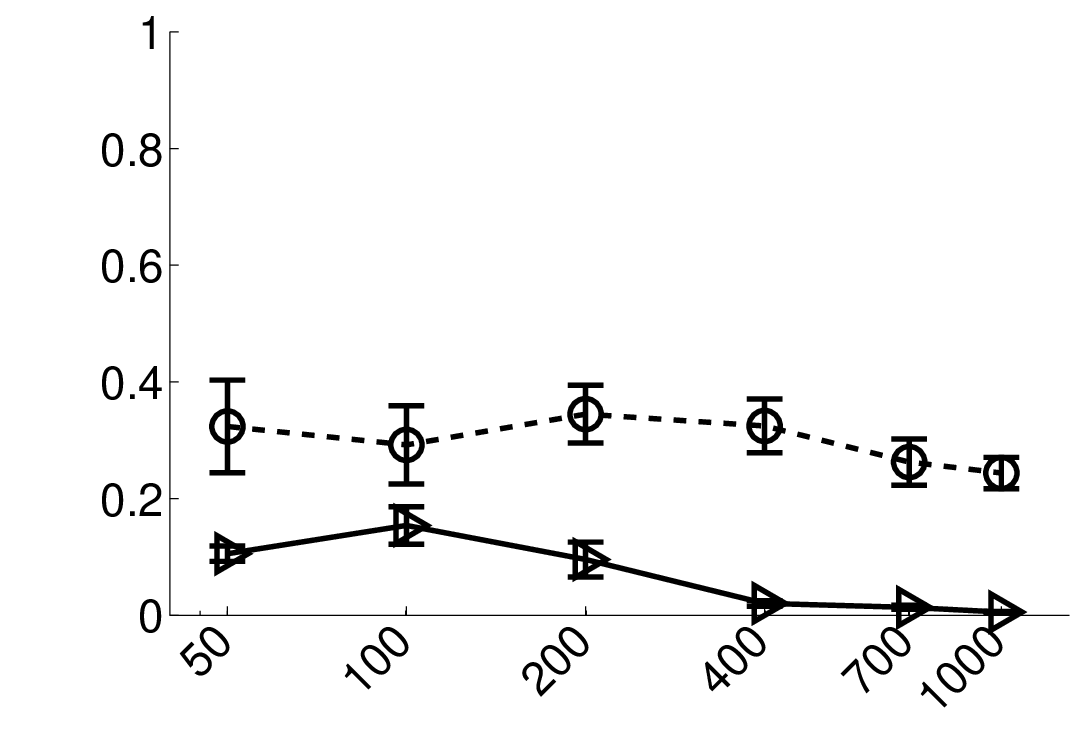} &
\includegraphics[width=1.5in]{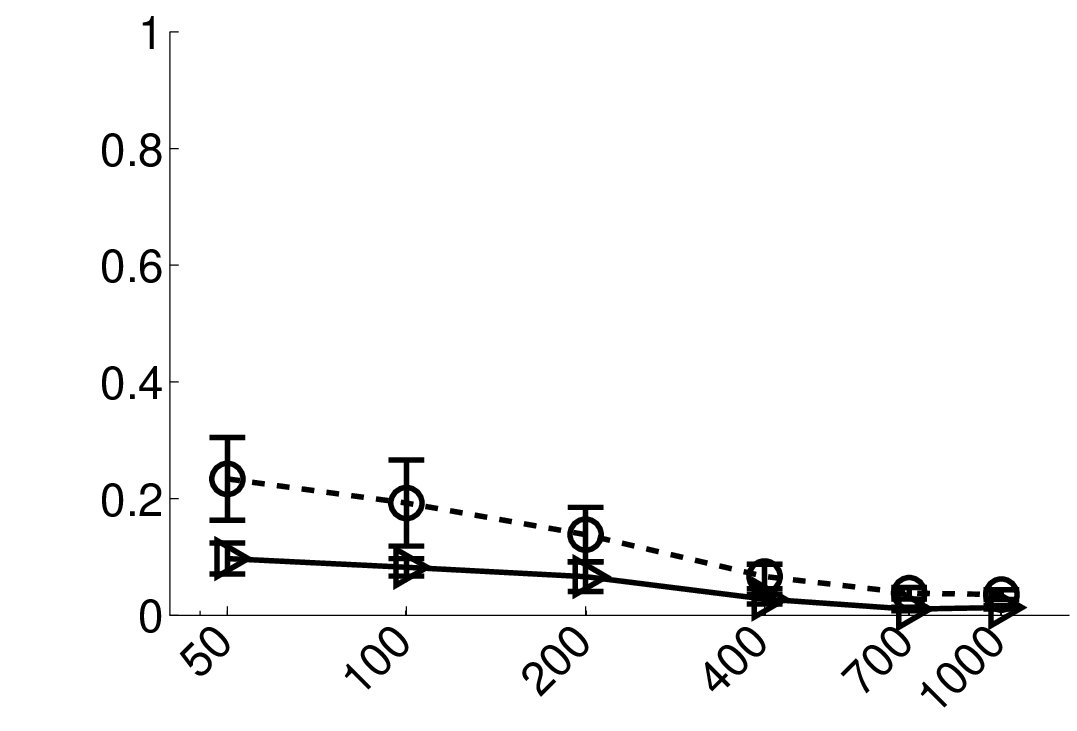} &
\includegraphics[width=1.5in]{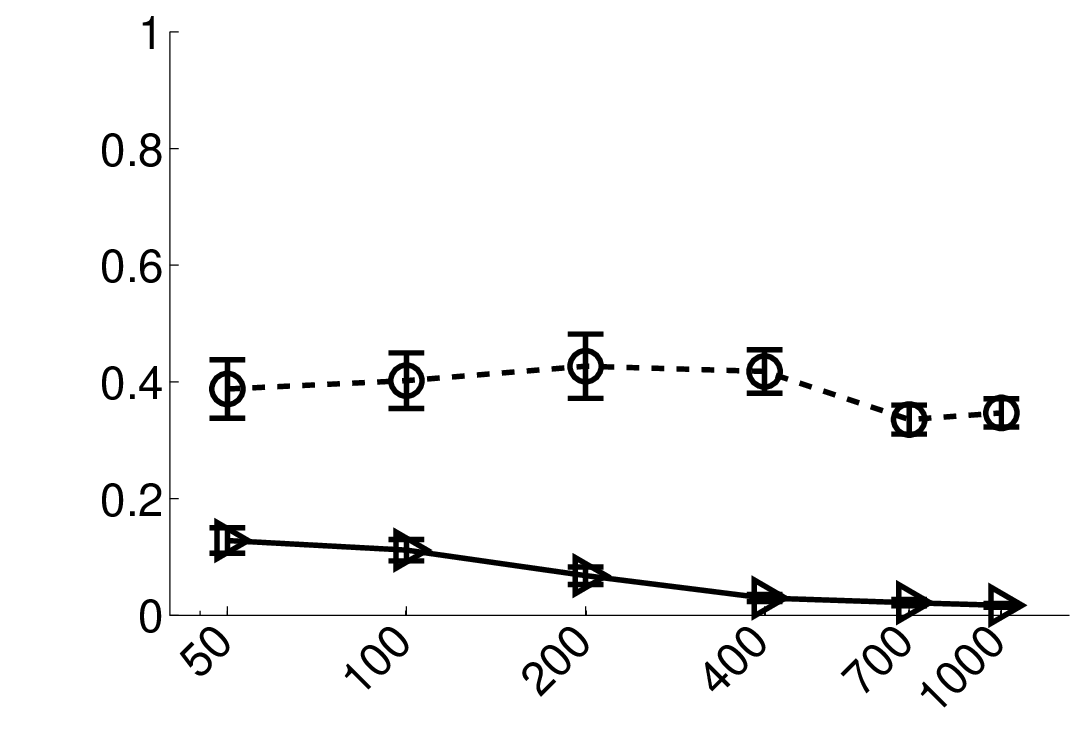}\\
(d) &
\includegraphics[width=1.5in]{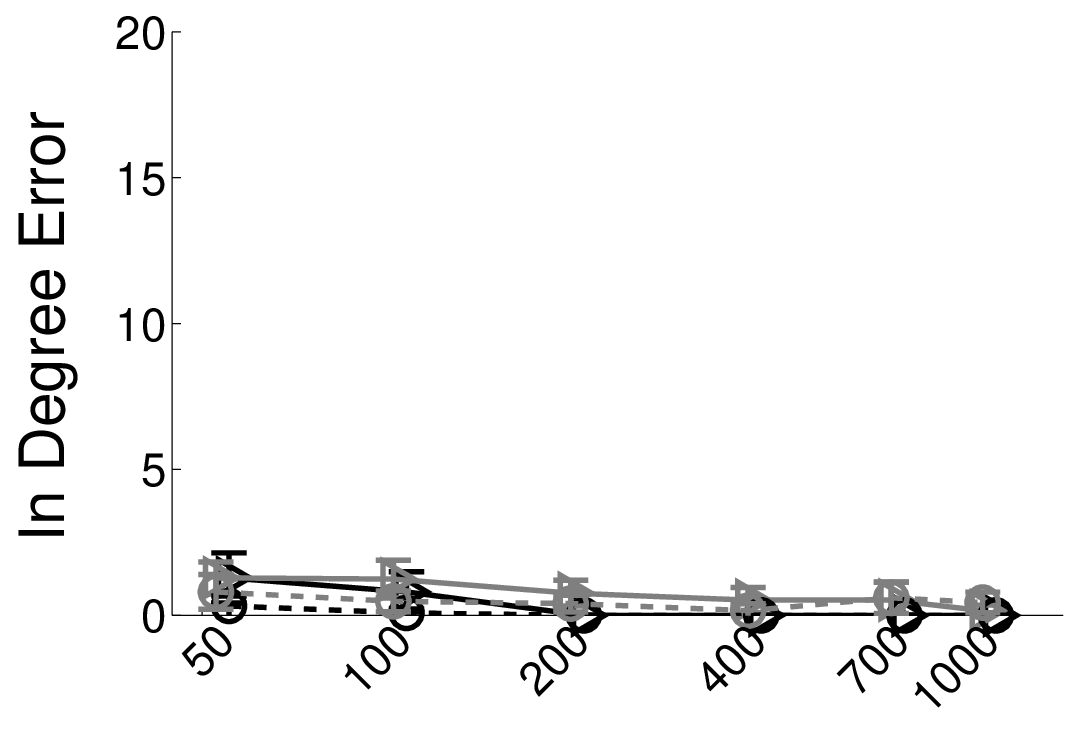} &
\includegraphics[width=1.5in]{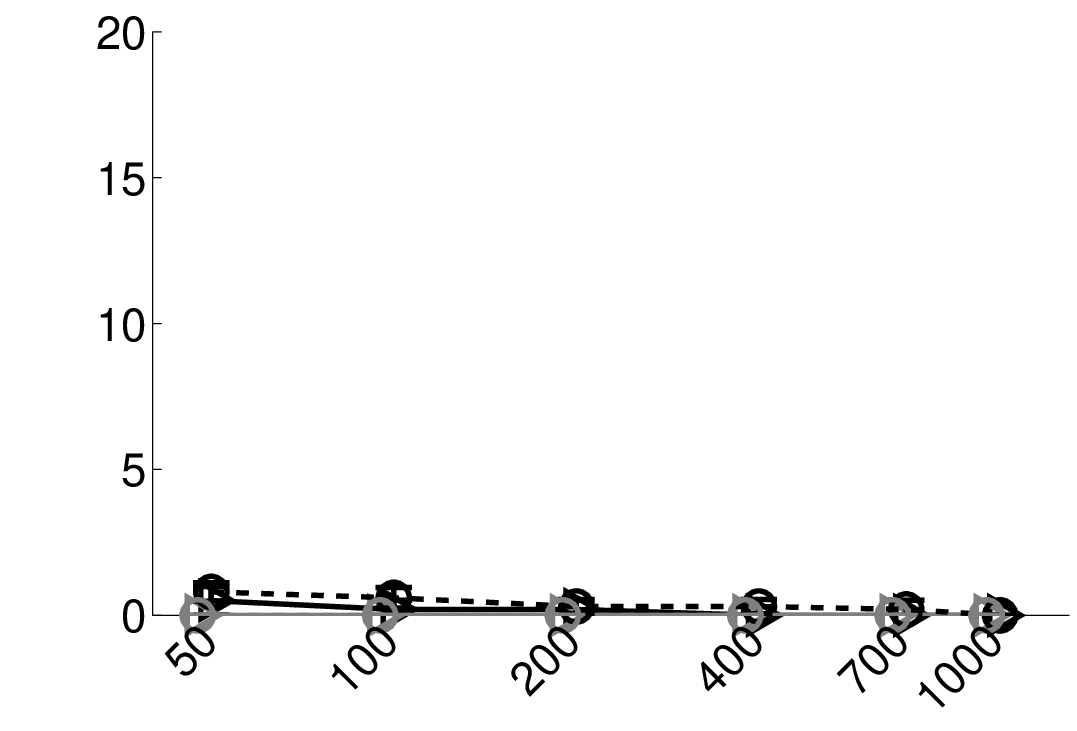} &
\includegraphics[width=1.5in]{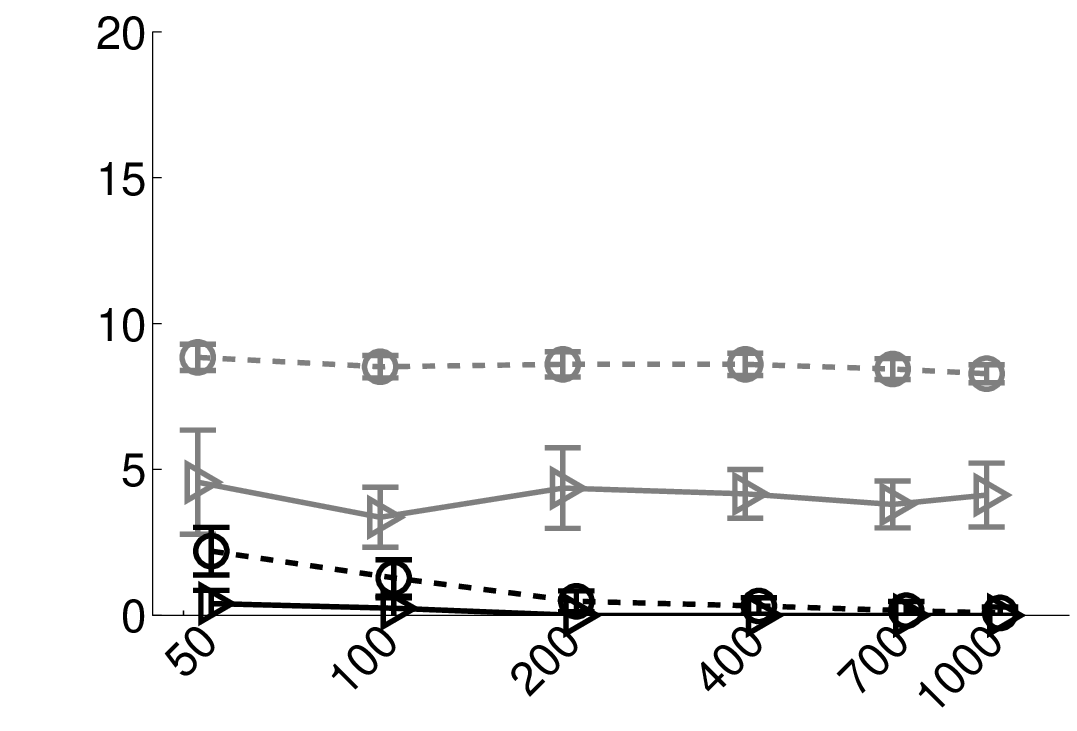} &
\includegraphics[width=1.5in]{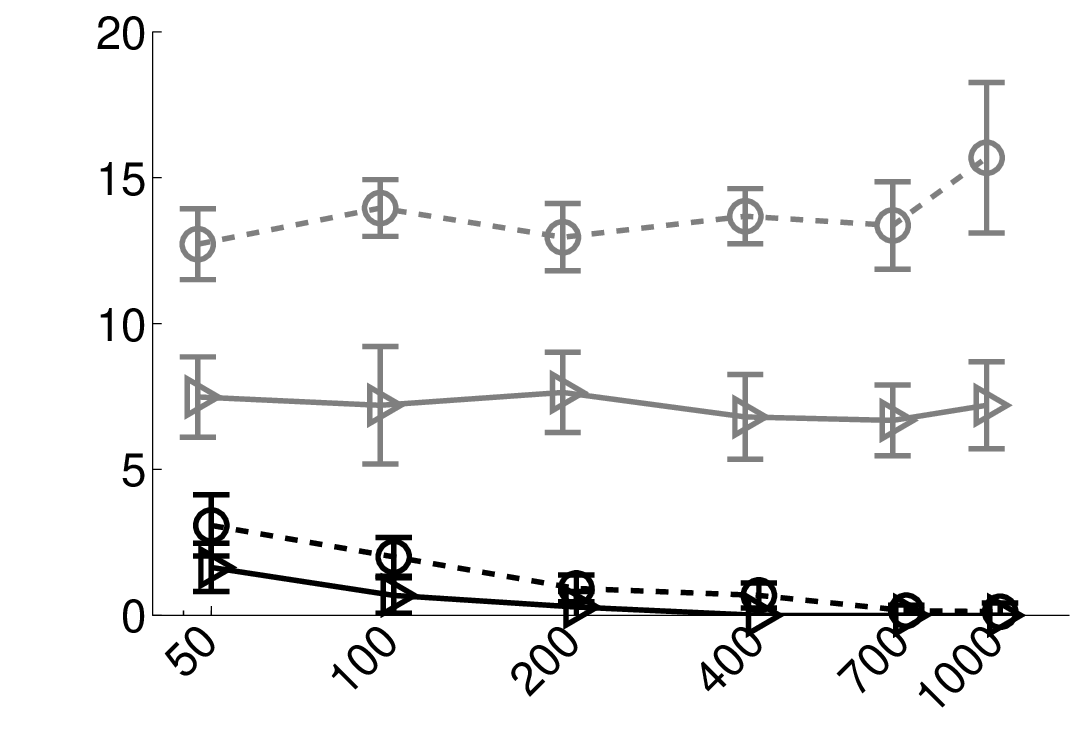}\\
(e) &
\includegraphics[width=1.5in]{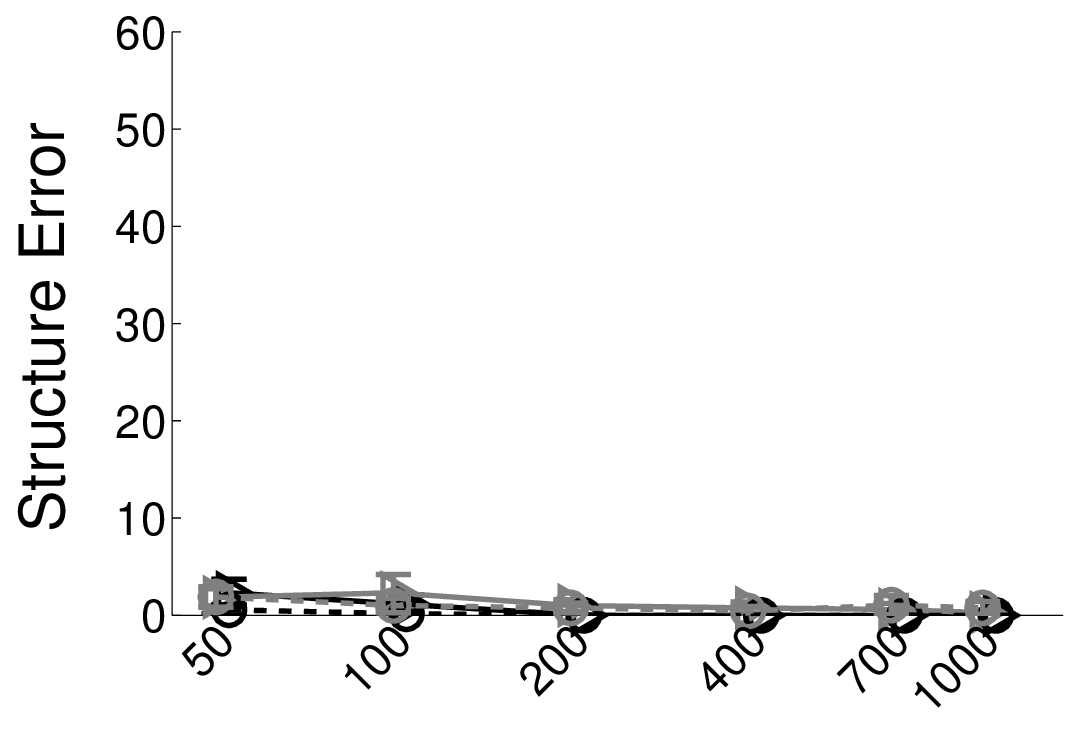} &
\includegraphics[width=1.5in]{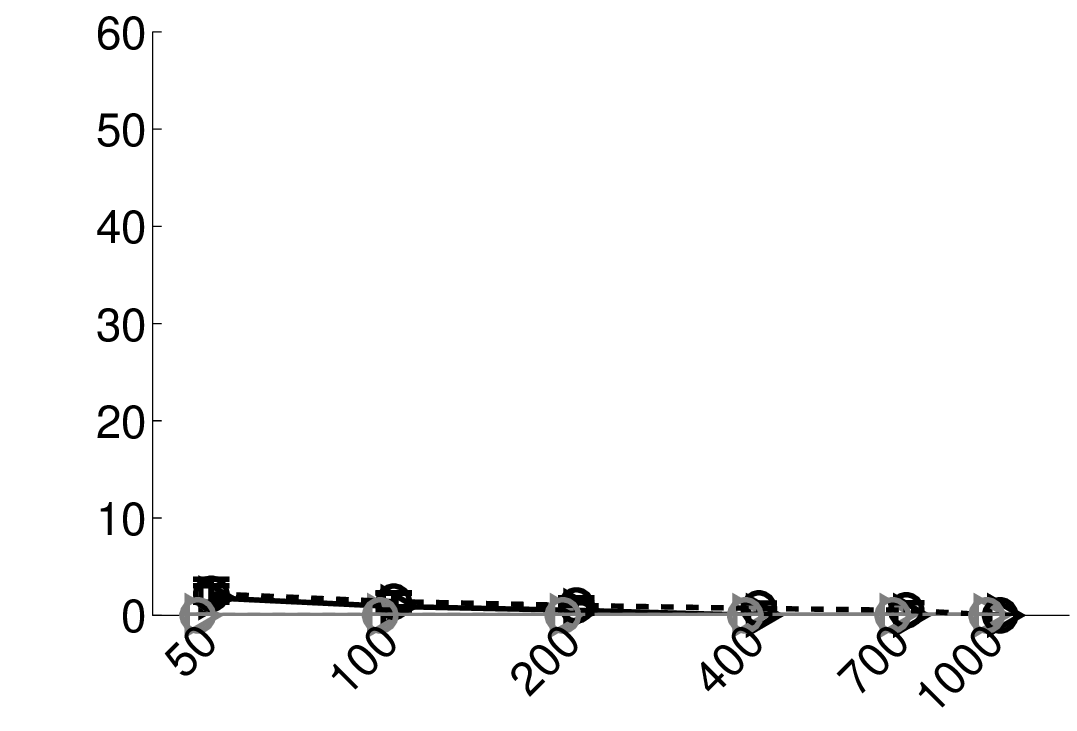} &
\includegraphics[width=1.5in]{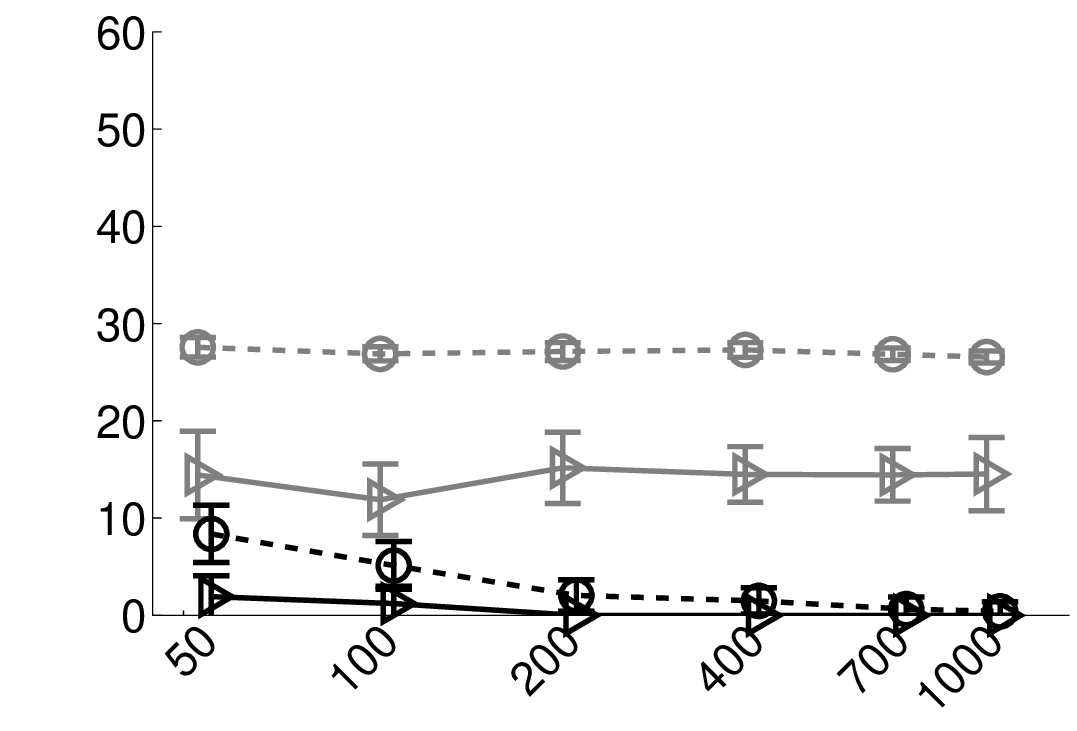} &
\includegraphics[width=1.5in]{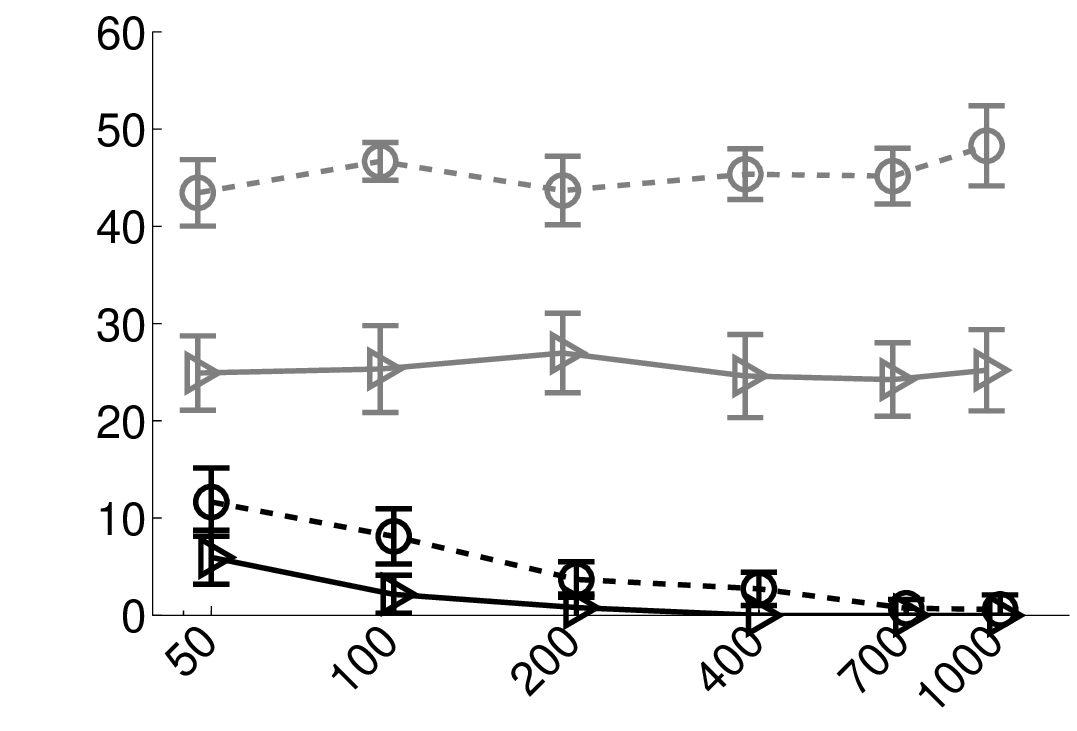}\\
(f) &
\includegraphics[width=1.5in]{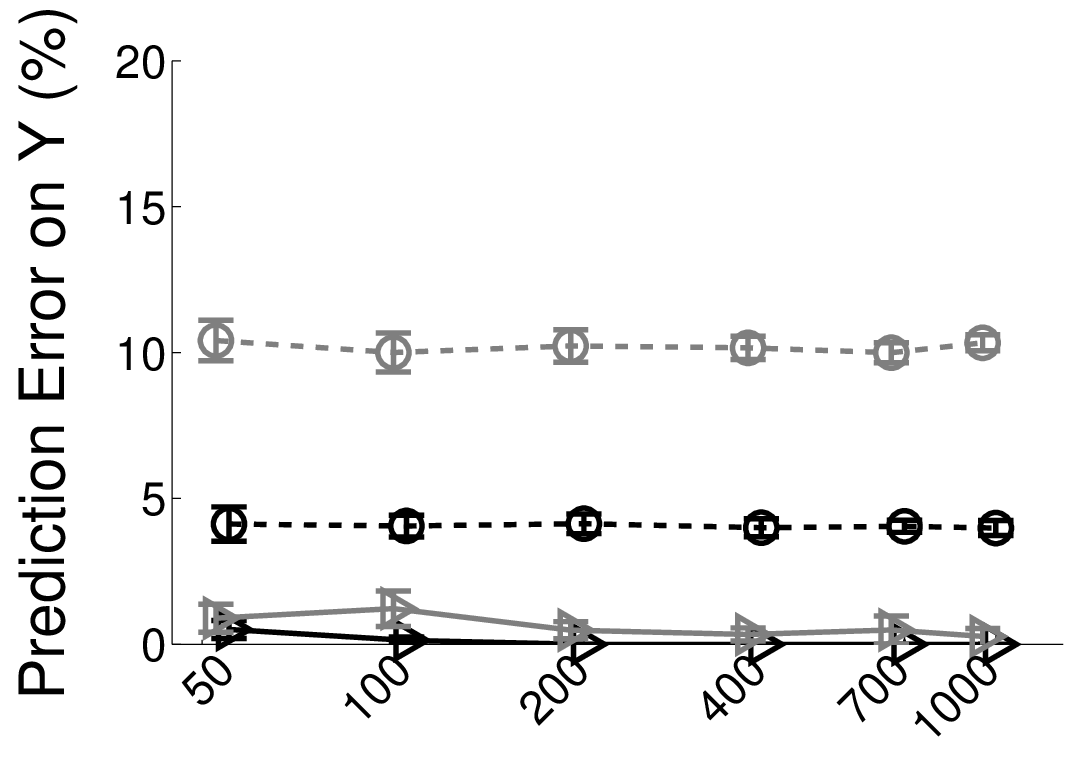} &
\includegraphics[width=1.5in]{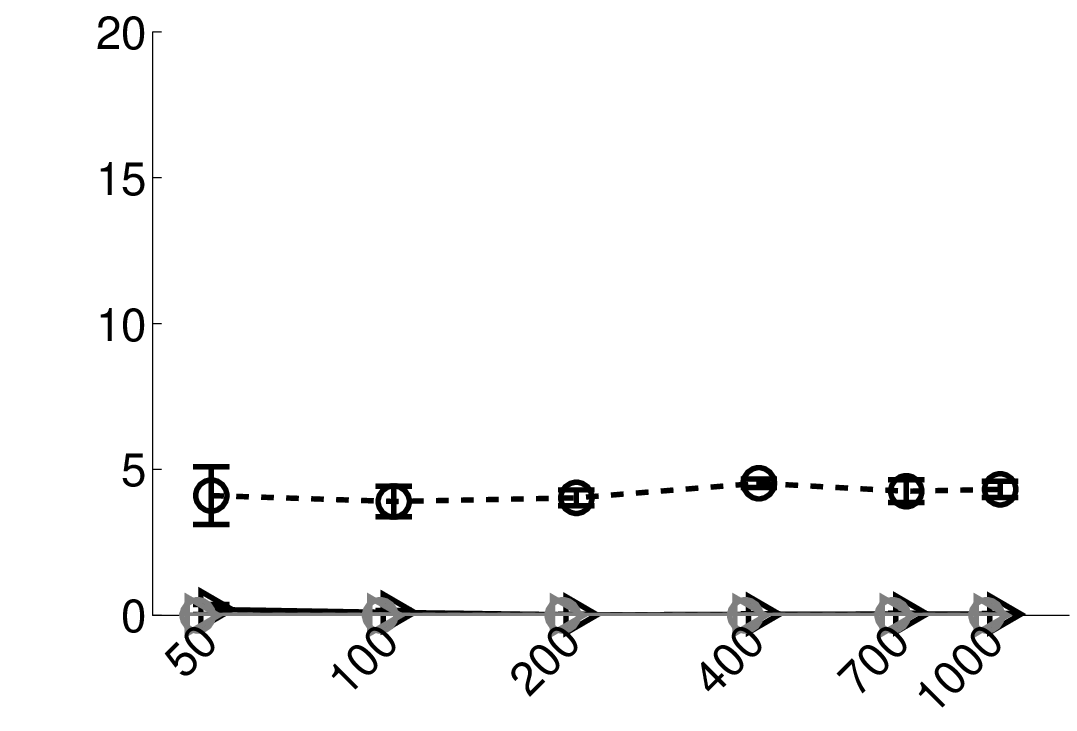} &
\includegraphics[width=1.5in]{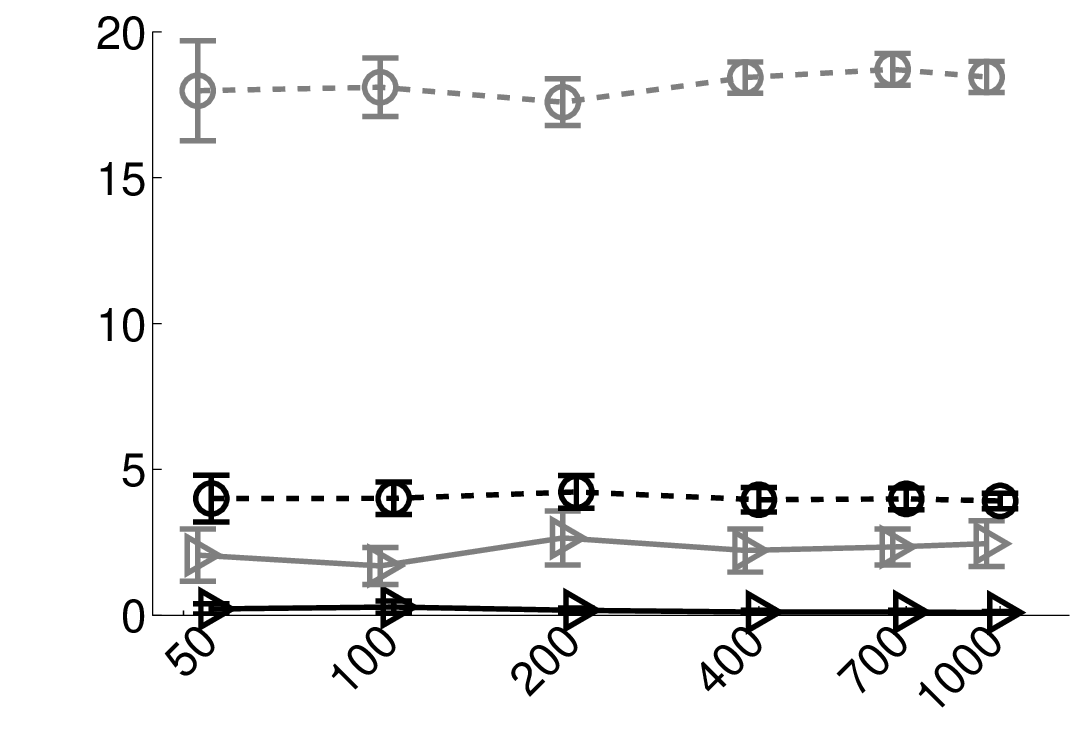} &
\includegraphics[width=1.5in]{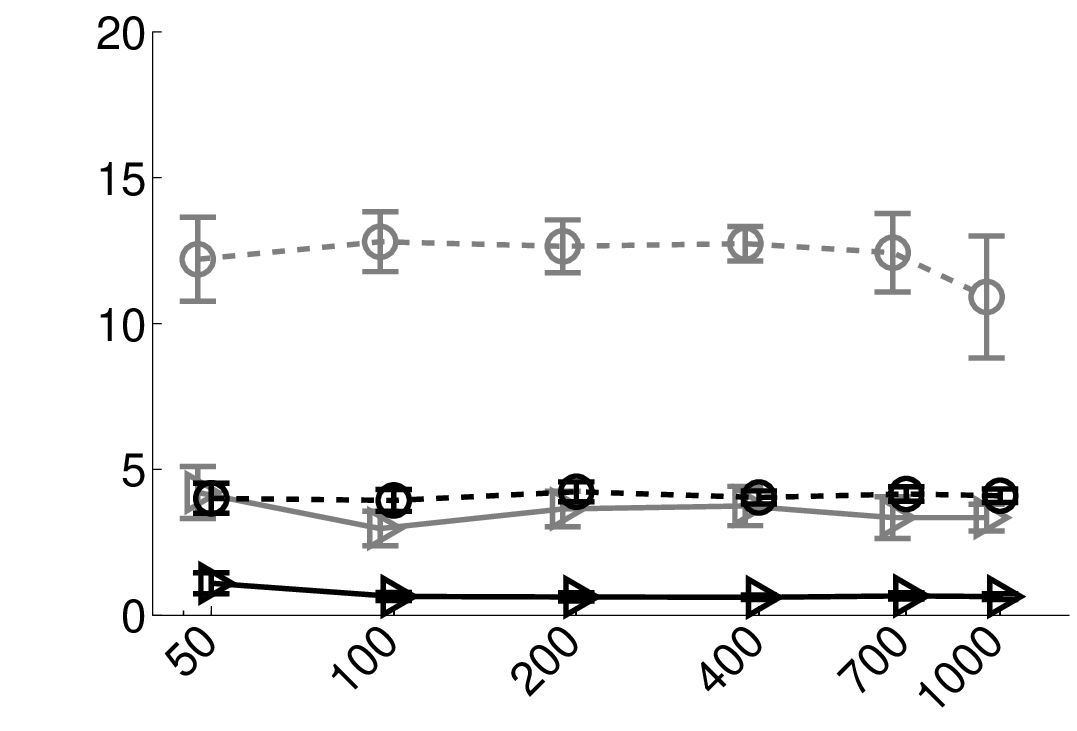}\\
(g) &
\includegraphics[width=1.5in]{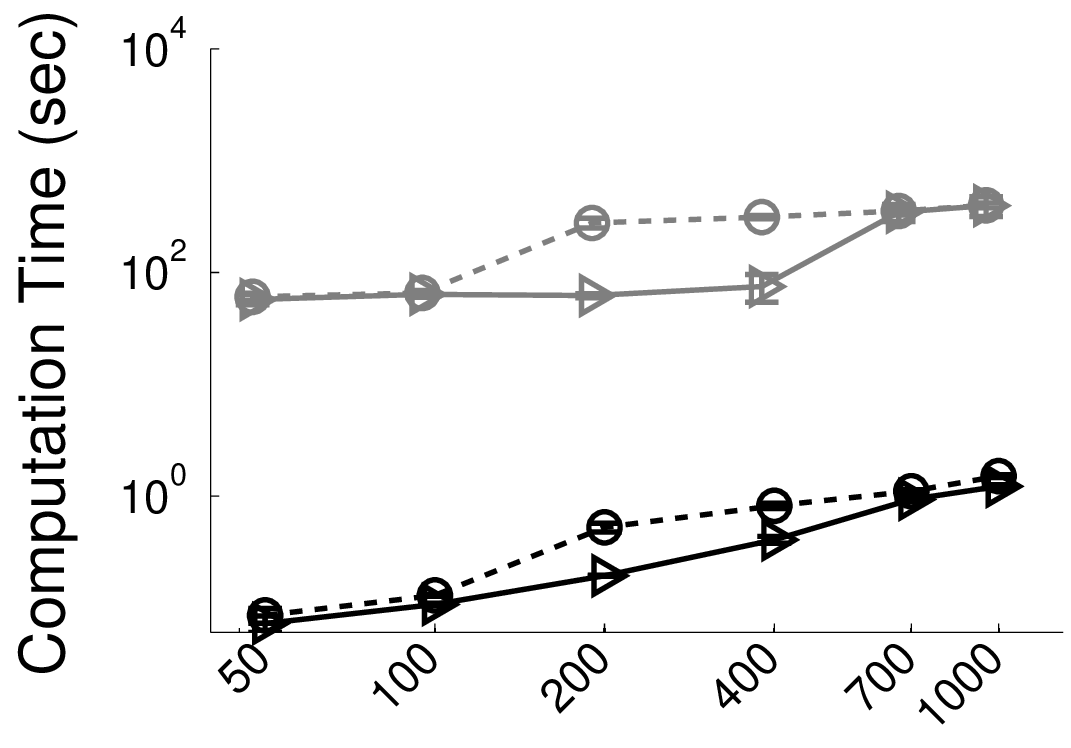} &
\includegraphics[width=1.5in]{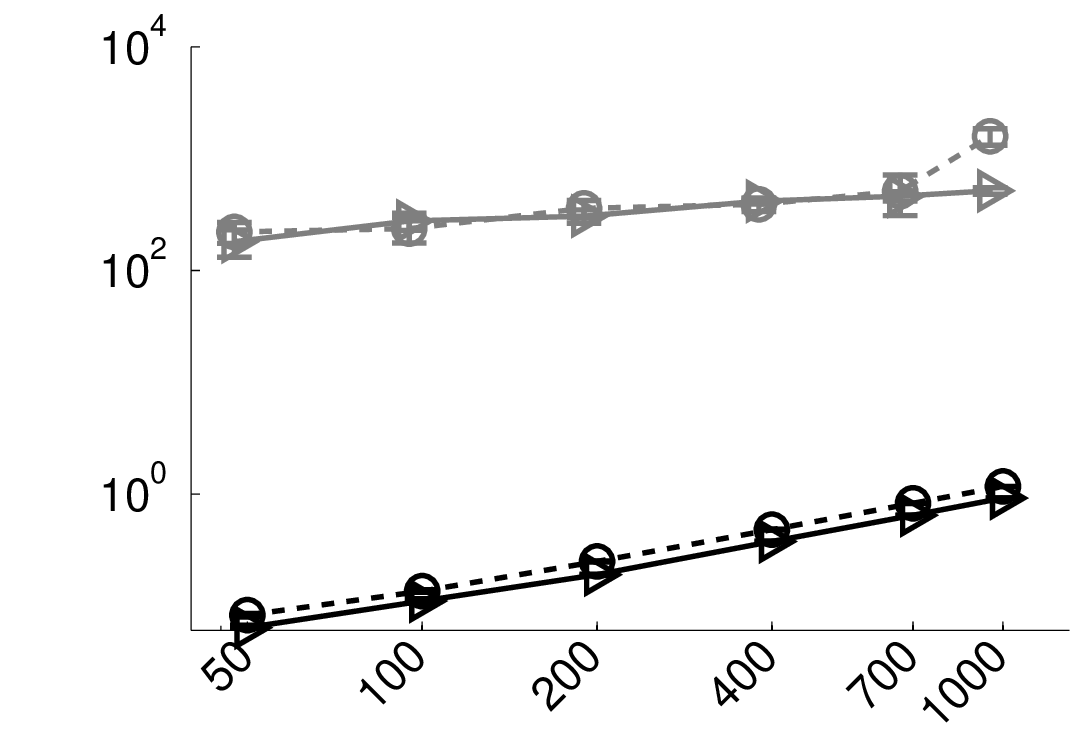} &
\includegraphics[width=1.5in]{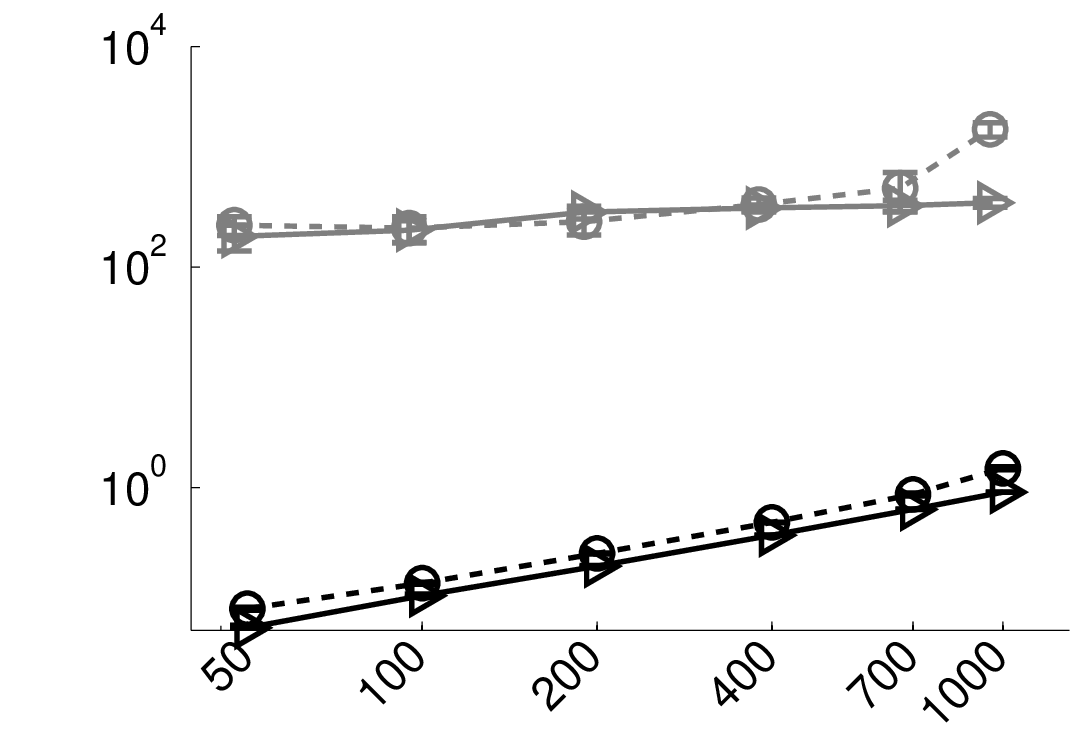} &
\includegraphics[width=1.5in]{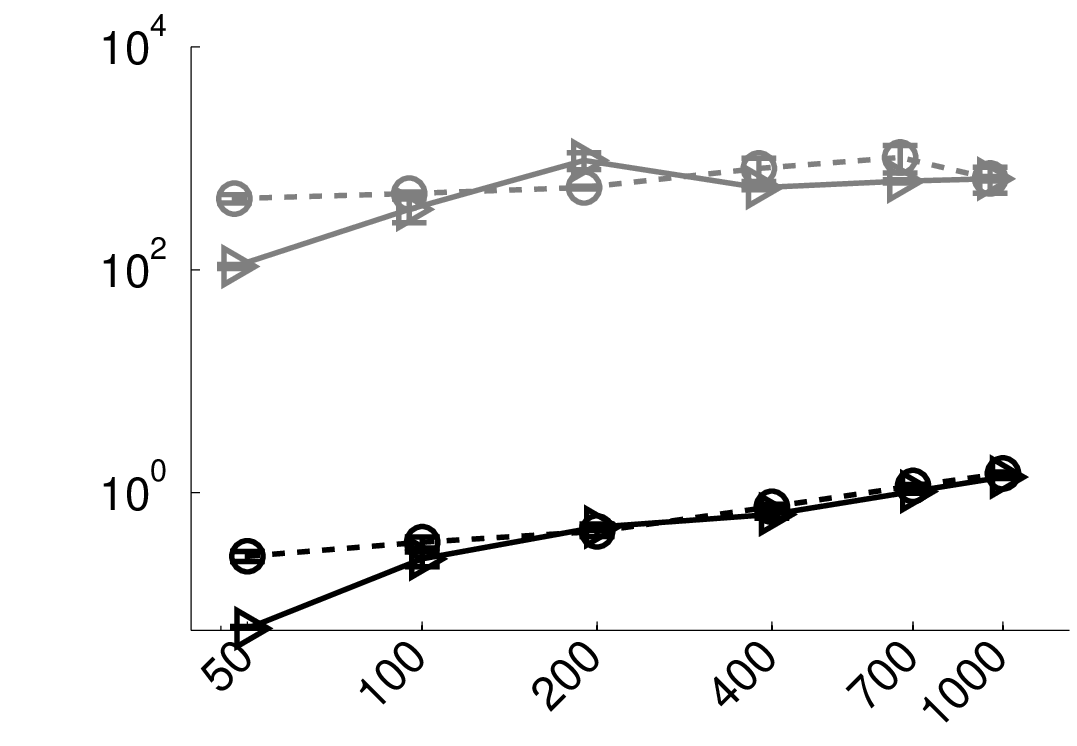}
\end{tabular}
\begin{tabular}{ c c c c c }
\hspace{0.1in} & \hspace{0.7in}(i)\hspace{0.7in} & \hspace{0.7in}(ii)\hspace{0.7in} & \hspace{0.7in}(iii)\hspace{0.7in} & \hspace{0.7in}(iv)\hspace{0.7in}
\end{tabular}
\end{center}
\caption{Experiment result for bICA and MAC with fixed network topologies.
Black lines are bICA performance while gray lines are MAC performance.
The top row shows the four network topologies. The second to fifth rows show the
mean Normalized Hamming Distance, Transmission Probability Error, In Degree Error,
Structure Error and Activity Error Ratio of
four topologies as $T$ increases from 50 to 1,000. On the last
row is mean CPU runtime measured in seconds. Each graph shows
experiment results at 2 different levels of noise: 0\% and 10\%.
Error bars are symmetric, and indicate standard deviation
over 50 runs with different initialized seeds.} \label{fig:eval}
\end{figure*}
We have implemented the proposed algorithm in Matlab.  Four network topologies
are manually created (Fig. \ref{fig:eval}(a)) representing different scenarios:
connected vs. disconnected network, under-determined ($m > n$) vs.
over-determined network ($m < n$). All experiments are conducted on a
workstation with an Intel Core 2 Duo T5750\makeatletter@\makeatother2.00GHz
processor and 2GB RAM. For each scenario, 50 runs are executed; the results
presented are the average value and 99\% confidence interval.  To evaluate the
robustness of the interference algorithm, two different levels of noise are
introduced, i.e. $p_e$ = 0\% and 10\%. In our experiment, noise is generated by
randomly flipping an entry of the observation matrix $\X$ with probability
$p_e$.

Evaluation results for bICA and MAC at the two noise levels are presented in
Fig.~\ref{fig:eval}(b), (d), (e), and (f). We do not include the results of MAC
in Fig.~\ref{fig:eval}(c) since it does not infer the active probability for
hidden components.  From Fig.~\ref{fig:eval}(b), we observe at zero noise, bICA
converges quickly to the ground truth, and $\G$ matrix has been accurately
estimated with only 100 observations. bICA is comparable to or slightly
outperforms MAC in the first two topologies, and significantly outperform  MAC
on the later two.  It appears the MAC is quite sensitive to noise even in small
structures. In contrast, the accuracy of bICA under 10\% noise degrades
when the number of the observations is small but improves significantly as more observations are available.
Thus, bICA is more resilient to noise than MAC.  As shown in
Fig.~\ref{fig:eval}(d) and (e), inference errors tend to increase with the same
number of observations for both schemes as the structure become more complex
for both schemes. However, bICA only degrades gracefully.
%

Fig.~\ref{fig:eval}(f) shows the accuracy of the solution to the inverse
problem.  In this set of experiments, we first determine $\hat{\G}$ and
$\hat{p}$ from the measurement data of length $T$. Then realizations of the
hidden sources are estimated by solving the MLE problem in \eqref{eq:IPP}. The
predication error is measured by the Activity Error Ratio.  We see that at the
0\% noise level, both methods perform quite well. Since the inference of each
realization of $\yy$ is independent from the others, increasing $T$ does not
help improving the accuracy of the inverse problem (though higher T gives a
better estimation of $\G$ and $p$.  Performance of both methods degrades as the
noise level increases. Noise has two effects on the solution to inverse
problem. First, the inferred mixing matrix and the active probability can be
erroneous. Second, no maximum likelihood estimator guarantees to give the exact
result when the problem is under or close to under-determined with noisy
measurements. To verify the second argument, we provide the exact $\G$ and $p$
to the MLE formulation, and observe comparable errors in the inferred $\yy$ as
the case where $\G$ and $p$ are both inferred by bICA.  This implies that the
main source of errors in the inverse problem comes from the under-determined or
close to under-determined name of the problem.

\begin{table}[h!]
\begin{center}
\caption{Average computation time (in seconds) of bICA and MAC over the 4 network topologies, 2 noise levels, and all sample sizes.}
\begin{tabular}{|c||c|c|c|c|}
\hline
 & Topology 1 & Topology 2 & Topology 3 & Topology 4\\
\hline \hline
bICA & 0.49 & 0.38 & 0.38 & 0.64 \\
\hline
MAC & 166.5 & 289.21 & 302.35 & 538.47 \\
\hline
\end{tabular}
\label{tab:time}
\end{center}
\end{table}

Finally, from Fig. \ref{fig:eval}(g), the computation time of bICA
is negligible (under 0.5 second in most cases and under 1.5 seconds
in the worst case). For the ease of comparison, we list the
numerical values in Table~\ref{tab:time}. MAC uses a gradient descent
optimization scheme, it is much more time-consuming.  As mentioned earlier, the
complexity of bICA is a function of $m$ and the sparsity of $\G$. In practice,
computation time of bICA is also a monotonic function of the number of
observations $T$ and noise levels. As $T$ is getting larger, the storage and
computation complexity tends to grow.  When $T$ is small or the noise level is
high, the structure inferred may error on the higher complexity side, resulting
longer computation time.

\section{Applications}
\label{sec:application}
In this section, we present some case studies on  real-world application of
bICA.  In general, bICA can be applied to any problem that need identifying
hidden source signals from binary observations. The proposed method therefore
can find applications in many domains. In multi-assignment clustering
\cite{Streich09}, where boolean vectorial data can simultaneously belong to
multiple clusters, the binary data can be modeled as the disjunction of the
prototypes of all clusters the data item belongs to. In medical diagnosis, the
symptoms of patients are explained as the result of diseases that are not
themselves directly observable \cite{computer06anon-parametric}.  Multiple
diseases can exhibit similar symptoms. In the Internet tomography
\cite{Castro04networktomography}, losses on end-to-end paths can be attributed
to losses on different segments (e.g., edges) of the paths.  In all above
generic applications, the underlying data models can be viewed as disjunctions
of binary independent components (e.g., membership of a cluster, presence of a
disease, packet losses on a network edge, etc). Now we will introduce in detail
two specific network applications in which bICA has been effectively applied.

\subsection{Optimal monitoring for multi-channel wireless networks}
Passive monitoring is a technique where a dedicated set of hardware devices
called \textit{sniffers}, or monitors, are used to monitor activities in
wireless networks. These devices capture transmissions of wireless devices or
activities of interference sources in their vicinity, and store the information
in trace files, which can be analyzed distributively or at a central location.
Most operational networks operate over multiple channels, while a single radio
interface of a sniffer can only monitor one channel at a time. Thus, the
important question is to decide the sniffer-channel assignment to maximize the
total information (user transmitted packets) collected.

\paragraph*{Network model and optimal monitoring}
Consider a system of $m$ sniffers, and $n$ users, where each user $u$ operates
on one of $K$ channels, $c(u) \in \MK=\set{1,\ldots,K}$. The users can be
wireless (mesh) routers, access points or mobile users.  At any point in time,
a sniffer can only monitor packet transmissions over {\bf a single channel}. We
represent the relationship between users and sniffers using an undirected
bi-partite graph $G=(S,U,E)$, where $S$ is the set of sniffer nodes and $U$ is
the set of users. An edge $e=(s,u)$ exists between sniffer $s\in S$ and user $u
\in U$ if $s$ can capture the transmission from $u$.  If transmissions from a
user cannot be captured by any sniffer, the user is excluded from $G$. For
every vertex $v \in U\cup S$, we let $N(v)$ denote vertex $v$'s neighbors in
$G$. For users, their neighbors are sniffers, and vice versa. We will also
refer to $\G$ as the binary $m \times n$ adjacency matrix of graph $G$.
An example network with sniffers and users, the corresponding bipartite graph $G$,
and its matrix representation $\G$ are given in Figure~\ref{fig:mt}.

\begin{figure}[tp]
\begin{center}
\begin{tabular}{cc}
\includegraphics[width=1.6in]{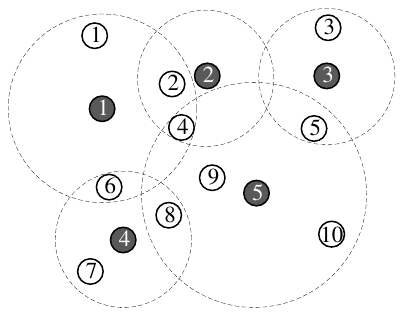} & \includegraphics[width=1.6in]{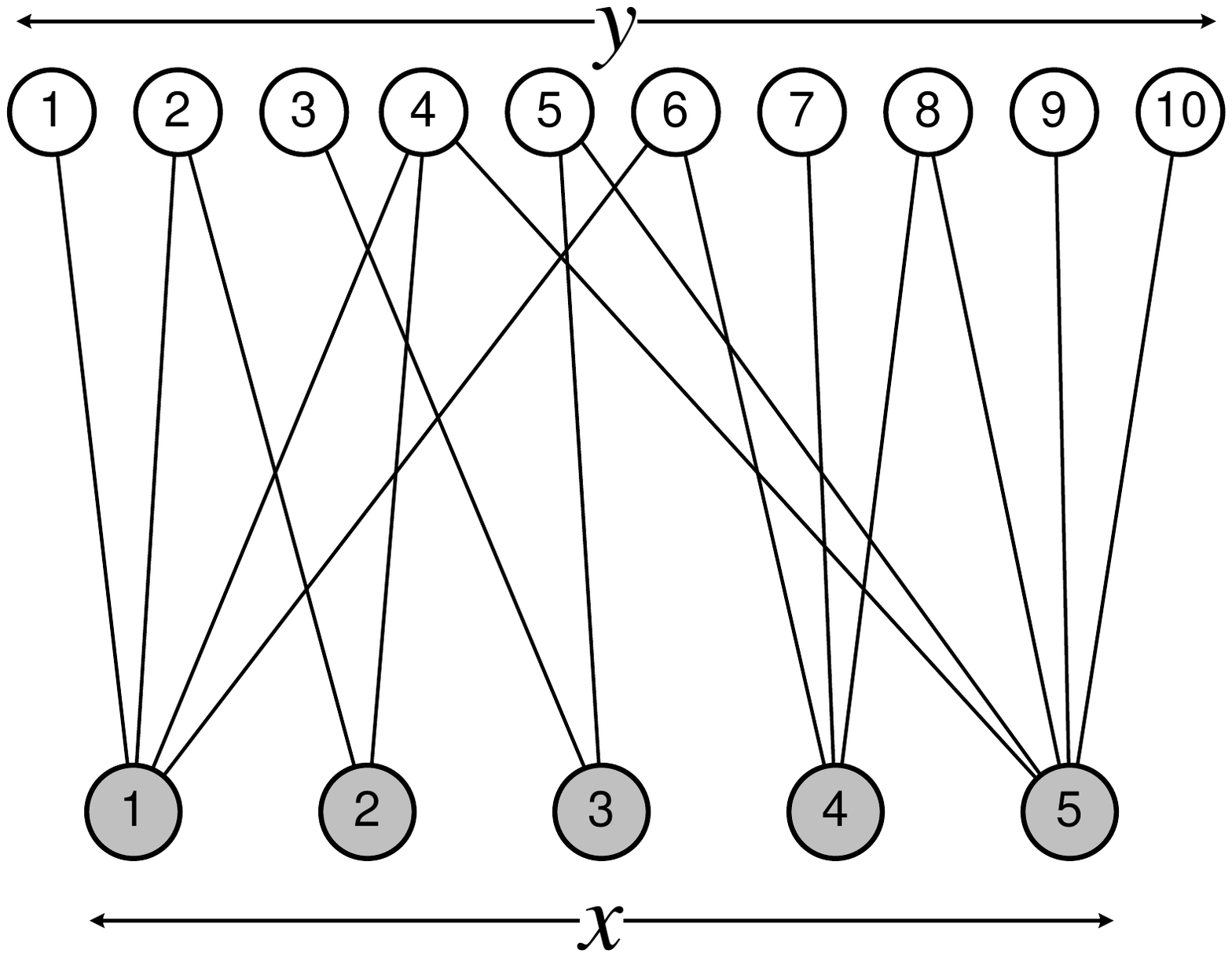}
\end{tabular}
$$
\G =
\begin{bmatrix}
1&\hspace{1.5mm}1&\hspace{1.5mm}0&\hspace{1.5mm}1&\hspace{1.5mm}0&\hspace{1.5mm}1&\hspace{1.5mm}0&\hspace{1.5mm}0&\hspace{1.5mm}0&\hspace{1.5mm}0\\
0&\hspace{1.5mm}1&\hspace{1.5mm}0&\hspace{1.5mm}1&\hspace{1.5mm}0&\hspace{1.5mm}0&\hspace{1.5mm}0&\hspace{1.5mm}0&\hspace{1.5mm}0&\hspace{1.5mm}0\\
0&\hspace{1.5mm}0&\hspace{1.5mm}1&\hspace{1.5mm}0&\hspace{1.5mm}1&\hspace{1.5mm}0&\hspace{1.5mm}0&\hspace{1.5mm}0&\hspace{1.5mm}0&\hspace{1.5mm}0\\
0&\hspace{1.5mm}0&\hspace{1.5mm}0&\hspace{1.5mm}0&\hspace{1.5mm}0&\hspace{1.5mm}1&\hspace{1.5mm}1&\hspace{1.5mm}1&\hspace{1.5mm}0&\hspace{1.5mm}0\\
0&\hspace{1.5mm}0&\hspace{1.5mm}0&\hspace{1.5mm}1&\hspace{1.5mm}1&\hspace{1.5mm}0&\hspace{1.5mm}0&\hspace{1.5mm}1&\hspace{1.5mm}1&\hspace{1.5mm}1
\end{bmatrix}
\begin{matrix}
x_1\\x_2\\x_3\\x_4\\x_5
\end{matrix}
\vspace{-2mm}
$$
$$
\hspace{5mm}
\begin{matrix}
y_1&y_2&y_3&y_4&y_5&y_6&y_7&y_8&y_9&y_{10}
\end{matrix}
$$
\end{center}
\caption{A sample network scenario with number of sniffers $m = 5$, number of users $n = 10$, its bipartite
graph transformation and its matrix representation. White circles represent independent users,
black circles represent sniffers and dashed lines illustrate sniffers' coverage range.}
\label{fig:mt}
\end{figure}

If we assume that $\G$ is known by inspecting packet headers information from each
sniffers' captured traces, then the transmission probability of the users
$p = (p_u)_{u\in U}$ are available and are assumed to be independent. As mentioned
earlier, the more complete information can be collected, the easier it is for a network
administrator to make decisions regarding network troubleshooting. We can
therefore measure the quality of monitoring by the total expected number of active users
monitored by the sniffers. Our problem now is to find a sniffer assignment of sniffers
to channels so that the expected number of active users monitored is maximized.
It can be casted as the following integer program:

\beq
\label{eq:ip-mec}
\begin{array}{rll}
\underset{y,z}{\max.} & \sum_{u \in U} p_u y_u \\[0.2cm]
\mbox{s.t.}
& \sum_{k=1}^K z_{s,k} \leq 1, & \quad \forall s\in S, \\
& y_u \leq  \sum_{s \in N(u)} z_{s,c(u)}, & \quad \forall u \in U,\\
& y_u \leq 1, & \quad \forall u\in U, \\
& y_u, z_{s,k} \in \set{0,1}, & \quad \forall u,s,k,
\end{array}
\eeq
where the binary decision variable $z_{s,k} = 1$ if the sniffer is assigned to channel $k$;
0 otherwise. $y_u$ is a binary variable indicating whether or not user $u$ is monitored,
and $p_u$ is the active probability of user $u$.

\paragraph*{Network topology inference with binary observations}
From (\ref{eq:ip-mec}), it is clear that we need the network and user-level
information in order to maximize the quality of monitoring. However, this
information is not always available.  We consider binary sniffers, or sniffers
that can only capture {\bf binary} information (\textit{on} or \textit{off})
regarding the channel activity. Examples of such kind of sniffers are energy
detection sensors using for spectrum sensing. The problem now is to infer the
user-sniffer relationship (i.e. $\G$) and the active probability of users from
the observation data (i.e. $\X$).  Let $\xx = [x_1, x_2,  \ldots, x_m]^T$ be a
vector of $m$ binary random variables and $\X$ be the collection of $T$
realizations of $\xx$, where $x_{it}$ denotes whether or not sniffer $s_i$
captures communication activities in its associated channel at time slot $t$.
Let $\yy = [y1, y2,\ldots, y_n]^T$ be a vector of $n$ binary random variables,
where $y_j = 1$ if user $u_j$ transmits in its associated channel, and $y_j =
0$ otherwise. Sniffer observations are thus disjunctive mixtures of user
activities. In other words, relationship between $\xx$ and $\yy$ is $\xx = \G
\otimes \yy$ and thus we  can use bICA to infer $\G$ and $p$.

\paragraph*{WiFi trace collection and evaluation result}
We evaluate our proposed scheme by data traces collected from the University of Houston campus
wireless network using 21 WiFi sniffers deployed in the Philip G. Hall.
Over a period of 6 hours, between 12 p.m. and 6 p.m., each sniffer captured approximately
300,000 MAC frames. Altogether, 655 unique users are observed operating over three
channels. The number of users observed on channels 1, 6, 11
are 382, 118, and 155, respectively. Most users are active less than 1\% of the time
except for a few heavy hitters. User-level information is removed leaving only
binary channel observation from each sniffer. $\G$ and $p$ are then inferred using
bICA and input to the integer program (\ref{eq:ip-mec}) to find the best sniffer channel
assignment that maximize the expected number of active users monitored. Obviously,
the more accurate the network model is inferred, the better the assignment is and the
more users are monitored. We also vary the result by randomly select a subset of sniffers
and observe the number of monitored users from this set of sniffers. Result is presented in
Fig.~\ref{fig:ton}

\begin{figure}[h]
\begin{center}
\includegraphics[width=3.0in]{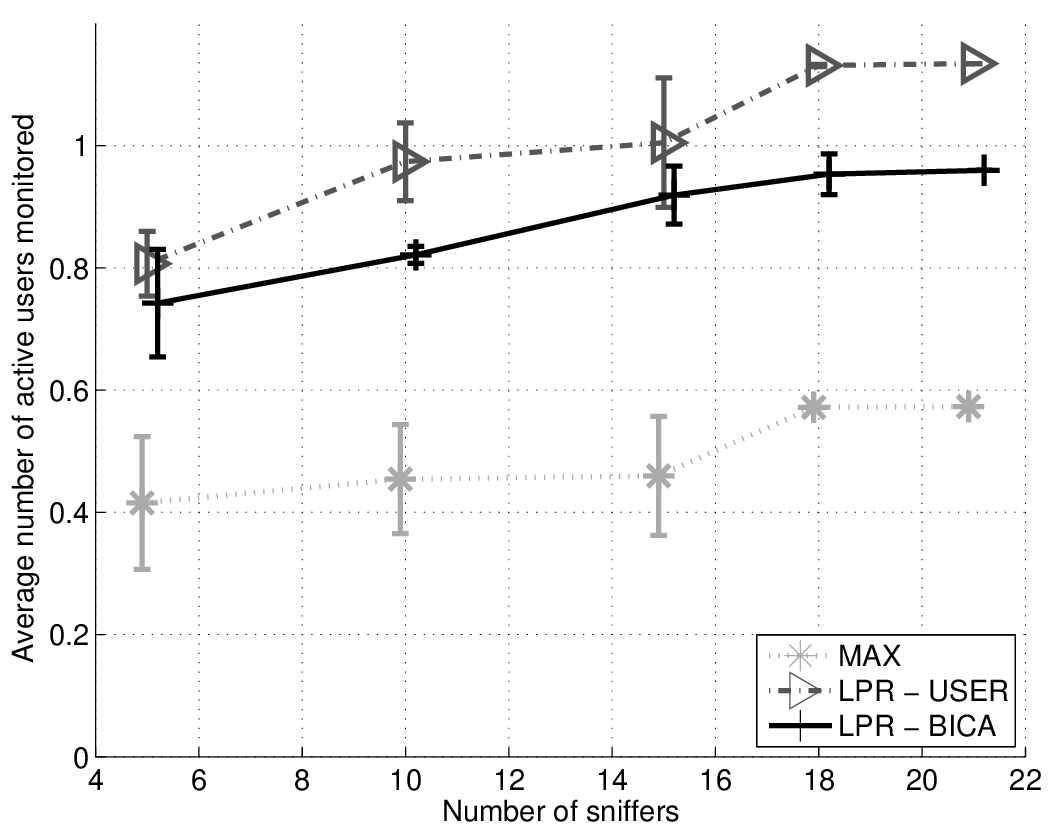}
\end{center}
\caption{Expected number of active users monitored with the number of sniffers vary from 5 to 21.}
\label{fig:ton}
\end{figure}

In Fig.~\ref{fig:ton}, we compare the average of number of active users
monitored using the inferred $\hat{\G}$ and $\hat{p}$, and the ground truth. The
integer programming problem (\ref{eq:ip-mec}) is solved using a random rounding
procedure on its LP relaxation, which is shown to perform very close to the LP upper
bound in our earlier work~\cite{zheng09QOM}.

Note that most users are active less than 1\% and the average active
probability of users is 0.0014. The system consists of 655 unique users,
therefore the average number of active users monitored is around 1. For
comparison, we also include a naive scheme (Max) that puts each sniffer to its
busiest channel. Therefore, Max does not infer or utilize any structure
information. From Fig.~\ref{fig:ton}, we observe that the sniffer-channel
assignment scheme with bICA (LPR -- BICA) performs close to the case when full information is
available (LPR -- USER), and much better than an agnostic scheme such as Max (MAX). This
demonstrates that bICA can indeed recover useful structure information from the
observations.


\subsection{PU separation in cognitive radio networks}
With tremendous growth in wireless services, the demand for radio spectrum
has significantly increased. However, spectrum resources are scarce and
most of them have been already licensed to existing operators. Recent studies
have shown that despite claims of spectral scarcity, the actual licensed
spectrum remains unoccupied for long periods of time~\cite{FCC}. Thus,
cognitive radio (CR) systems have been proposed~\cite{CR00,CR01,CR02} in order to
efficiently exploit these spectral holes{, in which licensed primary users (PUs) are not present}. CRs or secondary users (SUs) are
wireless devices that can intelligently monitor and adapt to their environment,
hence, they are able to share the spectrum with the licensed PUs, operating when the PUs are idle.

One key challenge in CR systems is spectrum sensing, i.e., SUs attempt to
learn the environment and determine the presence and characteristics of PUs.
Energy detection is one of the most commonly used method for spectrum sensing, where
the detector computes the energy of the received signals and
compares it to a certain threshold value to decide whether the PU signal
is present or not. It has the advantage of short detection time
but suffers from low accuracy compared to feature-based approaches such as cyclostationary
detection~\cite{CR01,CR02}.  From the prospective of a CR system, it is often
insufficient to detect PU activities in a single SU's vicinity (``is there any
PU near me?").  Rather, it is important to determine the identity of PUs (``who
is there?") as well as the distribution of PUs in the field (``where are
they?"). We call these issues the {\it PU separation problem}. Clearly, PU separation
is a more challenging problem compared to node-level PU detection.

\paragraph*{Solving PU separation problem with bICA}
Consider a slotted system in which the transmission activities of $n$ PUs are modeled
as a set of independent binary variables $\yy$ with active probabilities $\MP(\yy)$.
The binary observations due to energy detection at the $m$ monitor nodes are modeled
as an $m$-dimension binary vector $\xx = [x_1, x_2, \ldots, x_m]^T$ with joint
distribution $\MP(\xx)$. It is assumed that presence of any active PU surrounding
of a monitor leads to positive detection. If we let a (unknown) binary mixing matrix $\G$
represents the relationship between the observable variables in $\xx$ and the latent
binary variables in $\yy = [y_1, y_2, \ldots, y_n]^T$, then we can write $\xx = \G \otimes \yy$.
The PU separation is therefore amenable to bICA.

\paragraph*{Inferring PU activities with the inverse problem}
Extracting multiple PUs activities from the OR mixture observations is a challenging
but important problem in cognitive radio networks. Interesting information, such as the
PU channel usage pattern can be inferred once $\Y$ is available. The SUs will then be able
to adopt better spectrum sensing and access strategies to exploit the spectrum holes more
effectively. Now suppose that we are given the observation matrix $\X$ and already estimated
the mixing matrix $\G$ and the active probabilities $\MP(\yy)$. Solving the inverse problem
gives the PU activity matrix $\Y$.

\paragraph*{Simulation setup and result}
In the simulation, 10 monitors and $n$ PUs
are deployed in a 500x500 square meter area. We fix the sample size $T =
10,000$ and vary the number of PUs from 5 to 20 to study its impact on the accuracy of our
method. Locations of PUs are chosen arbitrarily on the field. The PUs' transmit power levels
are fixed at 20mW, the noise floor is -95dbm, and the propagation loss factor is 3. The SNR
detection threshold for the monitors is set to be 5dB. PUs' activities are modeled as a
two-stage Markov chain with transition probabilities uniformly distributed over [0; 1].
A monitor reports the channel occupancy if any detectable PU is active.
Noise is introduced by randomly flip a bit in the
observation matrix $\X$ from 1 to 0 (and vice versa) at probability $e$. $e$ is set at 0\%, 2\%, and 5\%.
Prediction error on $\G$ and $\Y$ over 50 runs for each PU setting are shown in Fig.~\ref{fig:twc}.

\begin{figure}[h]
\begin{center}
\includegraphics[width=3.6in]{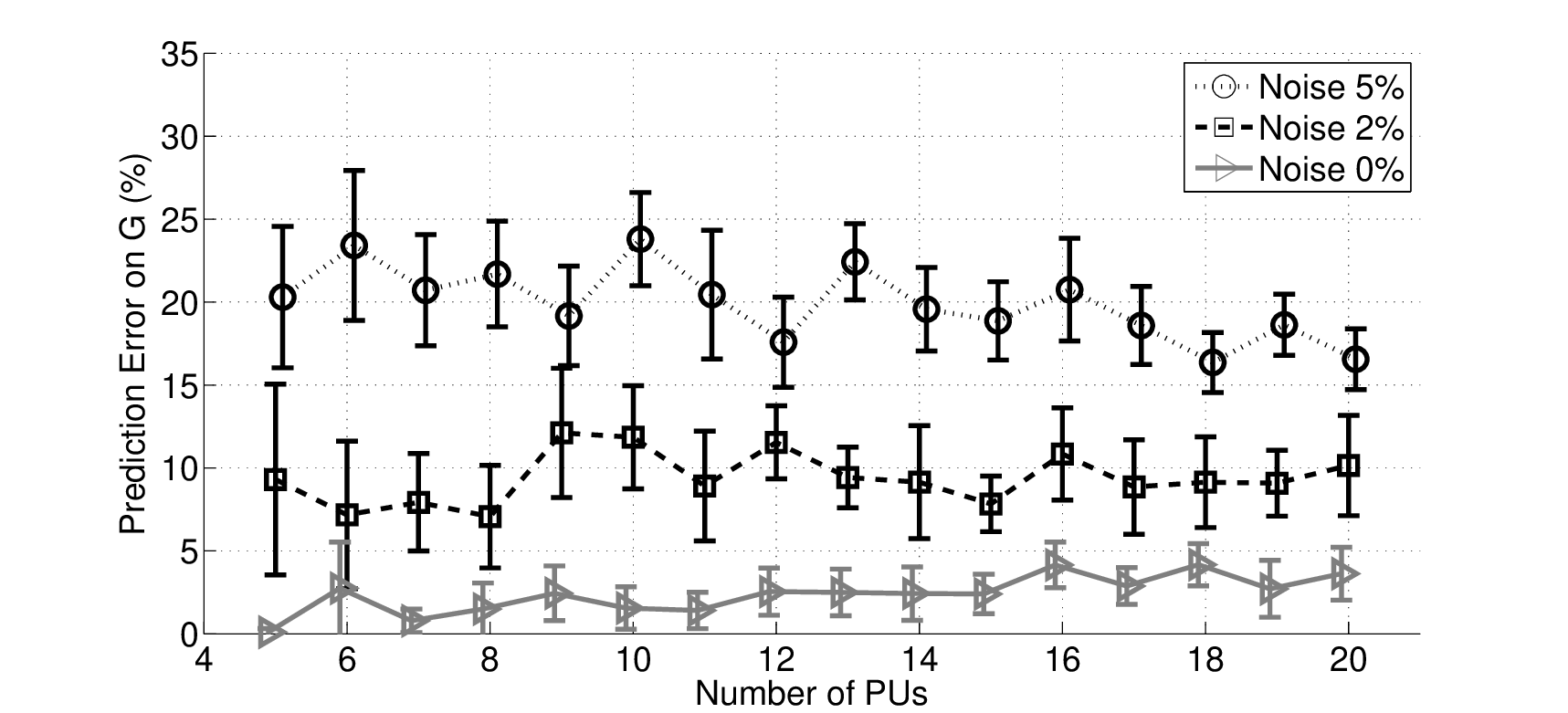} \\
\includegraphics[width=3.6in]{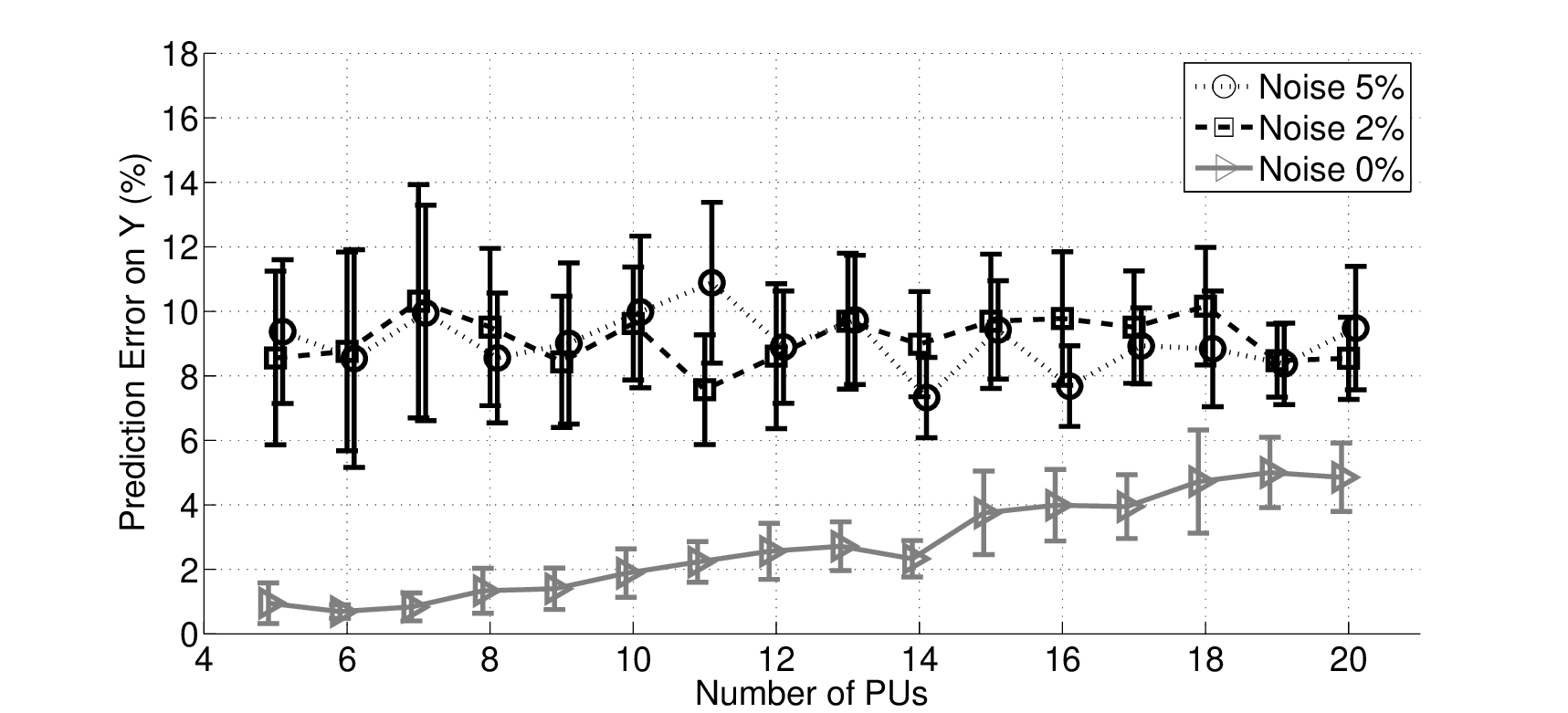}
\end{center}
\caption{Inference error on $\G$ and $\Y$ at three noise levels.}
\label{fig:twc}
\end{figure}

As we can see, at noise level 0\%, bICA can accurately estimate the underlying
PU-SU relationship and the hidden PU activities for small number of PUs.
However, introducing noise or having more PUs tend to degrade the performance of
the inference scheme. The errors in the noisy cases can be attributed to the
fact that the average PU active probability is around 2\%, which is comparable
to the noise level.  More information on this application can be found
in~\cite{zheng10twc}.

\section{Conclusion}
\label{sec:conclusion}
In this paper, we provided a comprehensive study of the binary independent
component analysis with OR mixtures. Key properties of bICA were established. A
computational efficient inference algorithm have been devised along with the
solution to the inverse problem. Compared to MAC, bICA is not only  faster, but
also more accurate and robust against noise.  We have also demonstrated the use
of bICA in two network applications, namely, optimal monitoring in
multi-channel wireless networks and PU separation in cognitive radio networks.
We believe the methodology devised can be useful in many other application
domains.

As future work, we are interested in devising inference schemes that can easily
incorporate priori knowledge of the structure or active probability of latent
variables. Also on the agenda is to apply bICA to problems in application
domains beyond wireless networks.

\section*{Acknowledgment}
The authors would like to thank Dr. Jian Shen from Texas State at San Marcos
for suggestions of the bipartite graph matching algorithm. The work is funded
in part by the National Science Foundation under award CNS-0832084.
\bibliographystyle{IEEEtran}

\end{document}